%% file: ERCE.tex
\def\IsDraft{} 

\documentclass[11pt]{article}
\usepackage{fullpage}

\usepackage{amsfonts,amsmath,amssymb,boxedminipage,color,url,nccmath,tikz}
\usepackage[bottom]{footmisc} 
\usepackage[pageanchor=false,pdfstartview=FitH,colorlinks,linkcolor=blue,filecolor=blue,citecolor=blue,urlcolor=blue]{hyperref} 

\usepackage{amsthm} 

\usepackage{booktabs}
\usepackage[T1]{fontenc}
\usepackage{lmodern} \normalfont 
\DeclareFontShape{T1}{lmr}{bx}{sc} { <-> ssub * cmr/bx/sc }{}
\usepackage{xcolor}
\usepackage{enumerate,paralist}
\usepackage{enumitem} 
\usepackage{varwidth}
\usepackage{aliascnt}
\usepackage[numbers]{natbib} 
\usepackage{cleveref} 
\usepackage{xspace}
\usepackage{xstring}
\usepackage[nottoc]{tocbibind} 
\usepackage[hypcap=false]{caption}
\usepackage{mathtools}
\usepackage[usestackEOL]{stackengine}
\setstackgap{L}{\normalbaselineskip}

\usetikzlibrary{fadings}
\usetikzlibrary{patterns}
\usetikzlibrary{shadows.blur}
\usetikzlibrary{shapes}

\usepackage[ruled,vlined,linesnumbered]{algorithm2e}

\input{macros}

\title{Communication Lower Bounds\\ for Cryptographic Broadcast Protocols\thanks{A preliminary version of this work appeared in \emph{DISC 2023}.}}

\author{Erica Blum\thanks{University of Maryland. E-mail: \texttt{erblum@umd.edu}.}
\and Elette Boyle\thanks{Reichman University and NTT Research. E-mail: \texttt{eboyle@alum.mit.edu}.}
\and Ran Cohen\thanks{Reichman University. E-mail: \texttt{cohenran@runi.ac.il}.}
\and Chen-Da Liu-Zhang\thanks{HSLU and Web3 Foundation. E-mail: \texttt{chen-da.liuzhang@hslu.ch}.}
}

\begin{document}
\sloppy

\maketitle
\thispagestyle{empty}

\input{abstract}

\Tableofcontents

\input{intro}

\input{preliminaries}

\input{lowerbound}
\input{lb_adaptive}

\input{flood}

\subsection*{Acknowledgments}
Part of E.\ Blum's work was done while the author was an intern at NTT Research.
E.\ Boyle's research is supported in part by AFOSR Award FA9550-21-1-0046 and ERC Project HSS (852952).
R.\ Cohen's research is supported in part by NSF grant no.\ 2055568 and by the Algorand Centres of Excellence programme managed by Algorand Foundation. Any opinions, findings, and conclusions or recommendations expressed in this material are those of the authors and do not necessarily reflect the views of Algorand Foundation.
Part of C.D.\ Liu-Zhang's work was done while the author was at NTT Research.

{\small{
\bibliographystyle{alpha}
\bibliography{crypto}
}}

\end{document}

%% file: macros.tex
\newcommand{\remove}[1]{}

\newcommand{\TLLNCS}[2]{\ifdefined\IsLLNCS#1\else #2 \fi}

\ifdefined\IsDraft
\newcommand{\authnote}[2]{[{\color{red}\textbf{#1:}}~{\color{blue} #2}]}
\else
\newcommand{\authnote}[2]{}
\fi

\newcommand{\Tableofcontents}{
\thispagestyle{empty}
\pagenumbering{gobble}
\clearpage
{\small{
\tableofcontents
}}
\thispagestyle{empty}
\clearpage
\pagenumbering{arabic}
}

\ifdefined\IsLLNCS

\else

\fi

\ifdefined\IsLLNCS

\else

\fi

\newcommand{\ie}  {i.e.,\xspace}
\newcommand{\eg}  {e.g.,\xspace}

\newcommand{\sset}[1]{\{#1\}}

\newcommand{\ssize}[1]{|#1|}

\renewcommand{\Pr}{{\mathrm {Pr}}}
\newcommand{\pr}[1]{\Pr\left[#1\right]}

\newcommand{\zo}{\{0,1\}}

\newcommand{\is}{{i^\ast}}
\newcommand{\Is}{{I^\ast}}

\newcommand{\Js}{{J^\ast}}

\newcommand{\poly}{\mathsf{poly}}
\newcommand{\polylog}{\mathsf{polylog}}

\newcommand{\assign}{\ensuremath{\mathrel{\vcenter{\baselineskip0.5ex \lineskiplimit0pt \hbox{\scriptsize.}\hbox{\scriptsize.}}}=}}

\renewcommand{\cref}{\Cref} 

\TLLNCS{
\newaliascnt{claiml}{theorem}
\newtheorem{claiml}[claiml]{Claim}
\aliascntresetthe{claiml}

\renewenvironment{claim}{\begin{claiml}}{\end{claiml}}
}
{
\newtheorem{theorem}{Theorem}[section]

\newaliascnt{lemma}{theorem}
\newtheorem{lemma}[lemma]{Lemma}
\aliascntresetthe{lemma}

\newaliascnt{claim}{theorem}
\newtheorem{claim}[claim]{Claim}
\aliascntresetthe{claim}

\newaliascnt{corollary}{theorem}

\aliascntresetthe{corollary}

\newaliascnt{proposition}{theorem}
\newtheorem{proposition}[proposition]{Proposition}
\aliascntresetthe{proposition}

\newaliascnt{conjecture}{theorem}

\aliascntresetthe{conjecture}

\newaliascnt{definition}{theorem}
\newtheorem{definition}[definition]{Definition}
\aliascntresetthe{definition}

\newaliascnt{remark}{theorem}

\aliascntresetthe{remark}

\newaliascnt{example}{theorem}

\aliascntresetthe{example}
}

\newaliascnt{construction}{theorem}

\aliascntresetthe{construction}

\crefname{lemma}{Lemma}{Lemmas}
\crefname{figure}{Figure}{Figures}
\crefname{claim}{Claim}{Claims}
\crefname{corollary}{Corollary}{Corollaries}
\crefname{proposition}{Proposition}{Propositions}
\crefname{conjecture}{Conjecture}{Conjectures}
\crefname{definition}{Definition}{Definitions}
\crefname{remark}{Remark}{Remarks}
\crefname{exmaple}{Example}{Examples}

\crefname{equation}{Equation}{Equations}

\newaliascnt{proto}{theorem}

\aliascntresetthe{proto}
\crefname{proto}{protocol}{protocols}

\newaliascnt{algo}{theorem}

\aliascntresetthe{algo}
\crefname{algo}{algorithm}{algorithms}

\newaliascnt{expr}{theorem}

\aliascntresetthe{expr}
\crefname{experiment}{experiment}{experiments}

\newcommand \mycaption {\small }     
\newcommand \mylabel {}

\ifdefined\IsLLNCS
    \newenvironment{nfbox}[3]{
    \renewcommand \mycaption {#1}
    \renewcommand \mylabel {#2}
    \begin{center}
    \begin{small}
    \begin{tabular}{|ll|}
    \hline
    \hspace{.3ex}
    \begin{minipage}{.97\linewidth}
         \vspace{0.5ex}
         #3
         }
         {
         \vspace{-2ex}
         \captionof{figure}{\mycaption}
         \label{\mylabel}
     \end{minipage}
     &\hspace{.3ex} \\
     \hline
     \end{tabular}
     \end{small}
     \end{center}
    }
\else
    
\fi

\newcommand{\Adv}{{\ensuremath{\mathsf{Adv}}}\xspace}

\newcommand{\Party}{{\mathsf{P}}\xspace}
\newcommand{\secParam}{{\ensuremath{\kappa}}\xspace}

\newcommand{\negl}{\mathsf{negl}}

\newcommand{\Ps}{{\Party_\is}\xspace}
\newcommand{\IS}{{I^\ast}\xspace}
\newcommand{\JS}{{J^\ast}\xspace}
\newcommand{\sender}{{\Party_s}\xspace}

\newcommand{\calS}{{\ensuremath{\mathcal{S}}}\xspace}
\newcommand{\calD}{{\ensuremath{\mathcal{D}}}\xspace}
\newcommand{\calA}{{\ensuremath{\mathcal{A}}}\xspace}
\newcommand{\calB}{{\ensuremath{\mathcal{B}}}\xspace}

\newcommand{\SMbox}[1]{\mbox{\scriptsize {\sc #1}}}
\newcommand{\VIEW}{\SMbox{VIEW}}

\newcommand{\viewmain}{\VIEW^\textsf{main}}
\newcommand{\outmain}{Y^\mathsf{main}}
\newcommand{\viewcrashA}{\VIEW^\mathsf{crash\mhyphen A}}
\newcommand{\outcrashA}{Y^\mathsf{crash\mhyphen A}}
\newcommand{\viewcrashB}{\VIEW^\mathsf{crash\mhyphen B}}
\newcommand{\outcrashB}{Y^\mathsf{crash\mhyphen B}}
\newcommand{\viewcrashS}{\VIEW^{\mathsf{crash\mhyphen S}_\beta}}
\newcommand{\outcrashS}{Y^{\mathsf{crash\mhyphen S}_\beta}}

\mathchardef\mhyphen="2D

\newcommand{\iith}[1] {$#1$\textsuperscript{th}\xspace}
\newcommand{\ith}           {\iith{i}}

\newlength{\protowidth}
\newcommand{\pprotocol}[5]{
{\begin{figure}[#4]
\begin{center}
\setlength{\protowidth}{0.97\textwidth}
\addtolength{\protowidth}{-2\intextsep}

\fbox{
        \small
        \hbox{\quad
        \begin{minipage}{\protowidth}
    \begin{center}
    {\bf #1}
    \end{center}
        #5
        \end{minipage}
        \quad}
        }
        \caption{\label{#3} #2}
\end{center}
\vspace{-4ex}
\end{figure}
} }

\newcommand{\KeyGen}{\mathsf{KeyGen}}
\newcommand{\Sign}{\mathsf{Sign}}
\newcommand{\Ver}{\mathsf{Ver}}
\newcommand{\sk}{\mathtt{sk}}
\newcommand{\pk}{\mathtt{pk}}

\newcommand{\Fmine}{\mathcal{F}_{\mathsf{mine}}}
\newcommand{\mine}{\mathtt{mine}}
\newcommand{\verify}{\mathtt{verify}}
\newcommand{\ext}{\mathsf{Ext}}

\newcommand{\fbc}{\mathsf{Flood}\text{-}\mathsf{BC}_{\Fmine}^{n,\epsilon,\secParam}}

\newcommand{\flood}{\mathsf{Flood}}
\newcommand{\Flood}{\flood}
\newcommand{\floodrounds}{7\ln\big(\frac{n}{2\cdot (\ln(n)+\kappa)}\big)+2}

\newcommand{\exec}{\mathsf{Exec}}

\newcommand{\MC}{\mathsf{MC}}
\newcommand{\locality}{k}
\newcommand{\inputCoins}{\textsc{SetupAndCoins}}
\newcommand{\mainAttack}{\textsc{AttackMain}}
\newcommand{\secondAttack}{\textsc{AttackCrashA}}
\newcommand{\thirdAttack}{\textsc{AttackCrashB}}
\newcommand{\attackDisconnect}{\mathcal{E}^\mathsf{main}_\mathsf{disconnect}}
\newcommand{\crashADisconnect}{\mathcal{E}^\mathsf{crash \mhyphen A}_\mathsf{disconnect}}
\newcommand{\crashBDisconnect}{\mathcal{E}^\mathsf{crash \mhyphen B}_\mathsf{disconnect}}
\newcommand{\crashBStoP}{\mathcal{E}^\mathsf{crash \mhyphen B}_\mathsf{S\to P}}
\newcommand{\mainStoP}{\mathcal{E}^\mathsf{main}_\mathsf{S\to P}}
\newcommand{\crashAPtoS}{\mathcal{E}^\mathsf{crash \mhyphen A}_\mathsf{P\to S}}
\newcommand{\mainPtoS}{\mathcal{E}^\mathsf{main}_\mathsf{P\to S}}
\newcommand{\attackLowLocal}{\mathcal{E}^\mathsf{main}_\mathsf{low\mhyphen locality}}
\newcommand{\attackLowLocalCrash}{\mathcal{E}^{\mathsf{crashS}_\beta}_\mathsf{low\mhyphen locality}}
\newcommand{\AttackBeta}{\textsc{AttackCrashS}_\beta}

\newcommand{\Srv}{S}
\newcommand{\inputCoinsAdaptive}{\textsc{SetupAndCoins}}

%% file: abstract.tex

\begin{abstract}
Broadcast protocols enable a set of $n$ parties to agree on the input of a designated sender, even in the face of malicious parties who collude to attack the protocol. In the honest-majority setting, a fruitful line of work harnessed randomization and cryptography to achieve low-communication broadcast protocols with sub-quadratic total communication and with ``balanced'' sub-linear communication cost per party.

However, comparatively little is known in the \emph{dishonest-majority} setting. Here, the most communication-efficient constructions are based on the protocol of Dolev and Strong (SICOMP '83), and sub-quadratic broadcast has not been achieved even using randomization and cryptography. On the other hand, the only nontrivial $\omega(n)$ communication lower bounds are restricted to \emph{deterministic} protocols, or against \emph{strong adaptive} adversaries that can perform ``after the fact'' removal of messages.

\medskip
We provide new communication lower bounds in this space, which hold against arbitrary cryptography and setup assumptions, as well as a simple protocol showing near tightness of our first bound.

\begin{itemize}
\item {\bf Static adversary.}
We demonstrate a tradeoff between \emph{resiliency} and \emph{communication} for randomized protocols secure against $n-o(n)$ \emph{static} corruptions. For example, $\Omega(n\cdot \polylog(n))$ messages are needed when the number of honest parties is $n/\polylog(n)$; $\Omega(n\sqrt{n})$ messages are needed for $O(\sqrt{n})$ honest parties; and $\Omega(n^2)$ messages are needed for $O(1)$ honest parties.

Complementarily, we demonstrate broadcast with $O(n\cdot\polylog(n))$ total communication and balanced $\polylog(n)$ per-party cost, facing any constant fraction of static corruptions.

\item {\bf Weakly adaptive adversary.}
Our second bound considers $n/2 + k$ corruptions and a \emph{weakly adaptive} adversary that cannot remove messages ``after the fact.'' We show that any broadcast protocol within this setting can be attacked to force
an arbitrary party to send messages to $k$ other parties.

Our bound implies limitations on the feasibility of \emph{balanced} low-communication protocols: For example, ruling out broadcast facing $51\%$ corruptions, in which all non-sender parties have sublinear communication locality.
\end{itemize}
\end{abstract}

%% file: intro.tex

\section{Introduction}

In a \emph{broadcast} protocol (a.k.a.\ Byzantine generals~\cite{PSL80,LSP82}) a designated party (the sender) distributes its input message in a way that all honest parties agree on a common output, equal to the sender's message if the sender is honest. Broadcast is amongst the most widely studied problems in the context of distributed computing, and forms a fundamental building block in virtually any distributed system requiring reliability in the presence of faults. The focus of this work is on synchronous protocols that proceed in a round-by-round manner.

Understanding the required communication complexity of broadcast is the subject of a rich line of research. Most centrally, this is measured as the total number of bits communicated within the protocol, as a function of the number of parties $n$ and corrupted parties $t$.
Other metrics have also been studied, such as \emph{message complexity} (number of actual messages sent), \emph{communication locality} (defined as the \emph{maximal degree} of a party in the induced communication graph of the protocol execution \cite{BGT13}), and \emph{per-party} communication requirements (measuring how communication is split across parties).

The classical lower bound of Dolev and Reischuk \cite{DR85} showed that $\Omega(n+t^2)$ messages are necessary for \emph{deterministic} protocols (a cubic message complexity is also sufficient facing any linear number of corruptions \cite{DR85,DS83,MR21}\footnote{We note that \cite{DR85,DS83,MR21} rely on cryptography, so they are not deterministic per se; however, these protocols make a black-box use of the cryptographic primitives and are deterministic otherwise.}).
This result came as part of several seminal impossibility results for deterministic protocols presented in the '80s, concerning feasibility \cite{FLP83}, resiliency \cite{LSP82,FLM86}, round complexity \cite{FL82,DS83}, and connectivity \cite{Dolev82,FLM86}. Those lower bounds came hand in hand with feasibility results initiated by Ben-Or \cite{BenOr83} and Rabin \cite{Rabin83}, as well as Dolev and Strong \cite{DS83}, showing that randomization and cryptography are invaluable tools in achieving strong properties within broadcast protocols.

As opposed to bounds on feasibility, resiliency, and round complexity, the impossibility of \cite{DR85} held for over 20 years, in both the honest- and dishonest-majority settings. Recently, there has been progress in the honest-majority setting, with several works demonstrating how randomization and cryptography can be used to bypass the classical communication complexity bound and achieve sub-quadratic communication: with information-theoretic security \cite{KSSV06,KS09,KS11,BGH13} or with computational security under cryptographic assumptions \cite{CM19,ACDNPRS19,CKS20,BKLL20,BCG21}; some of these protocols even achieve poly-logarithmic locality and ``balanced'' sub-linear communication cost per party. While the security of some of these constructions holds against a \emph{static} adversary that specifies corruptions before the protocol's execution begins, some of these protocols are even secure against a \emph{weakly adaptive} adversary; that is, an adversary that cannot retract messages sent by a party upon corrupting that party. Abraham et al.\ \cite{ACDNPRS19} showed that this relaxation is inherent for sub-quadratic broadcast, even for randomized protocols, by demonstrating an $\Omega(t^2)$ communication lower bound in the presence of a \emph{strongly rushing} adversary; that is, an adversary that can ``drop'' messages by corrupting the sender after the message is sent but before it is delivered --- this ability is known as \emph{after-the-fact removal.}

Focusing on the \emph{dishonest-majority} setting, however, comparatively little is known about communication complexity.
Here, the most communication-efficient broadcast constructions are based on the protocol of Dolev and Strong~\cite{DS83}, and broadcast with $o(nt)$ messages has not been achieved even using randomization and cryptography.
The state-of-the-art protocols, for a constant fraction $t=\Theta(n)$ of corruptions, are due to Chan et al.\ \cite{CPS20} in the weakly adaptive setting under a trusted setup assumption, and to Tsimos et al.\ \cite{TLP22} in the static setting under a weaker setup assumption; however, both works require $\Omega(nt)$ communication, namely $\tilde{O}(n^2)$.\footnote{As standard in relevant literature, in this work $\tilde O$ notation hides polynomial factors in $\log(n)$ as well as in the cryptographic security parameter $\secParam$.}
On the other hand, the only nontrivial $\omega(n)$ communication lower bounds are those discussed above, restricted to deterministic protocols, or against strong adaptive adversaries.

\subsection{Our Contributions}

In this work, we explore the achievable communication complexity of broadcast in the dishonest-majority setting.
We provide new communication lower bounds in this space, which hold against arbitrary cryptographic and setup assumptions, as well as a simple protocol showing near tightness of our first bound.
Our results consider a \emph{synchronous} communication model: lower bounds in this model immediately translate into lower bounds in the \emph{asynchronous} and \emph{partially synchronous} models, whereas protocols in the latter models can only tolerate $t<n/3$ corruptions \cite{DLS88} implying that synchrony is inherently needed for our protocol construction.

\paragraph{Static adversary.}
We begin with the setting of static corruptions.
We demonstrate a simple modification to the protocol of Chan et al.\ \cite{CPS20}, incorporating techniques from Tsimos et al.\ \cite{TLP22}, which obtains a new protocol with essentially optimal $\tilde O(n)$ communication.
The resulting protocol relies on the same assumptions as \cite{CPS20}: namely, cryptographic verifiable random functions (VRFs)\footnote{A verifiable random function \cite{MRV99} is a pseudorandom function that provides a non-interactively verifiable proof for the correctness of its output.} and a trusted public-key infrastructure (PKI) setup, where the keys for each party are honestly generated and distributed. Further, the protocol is resilient against any constant fraction of static corruptions as in \cite{TLP22}, and achieves balanced $\tilde{O}(1)$ cost per party.

\def\ThmUBStatic
{
Let $0<\epsilon<1$ be any constant.
Assuming a trusted-PKI for VRFs and signatures, it is possible to compute broadcast with $\tilde O(n)$ total communication ($\tilde{O}(1)$ per party) facing a static adversary corrupting $(1-\epsilon)\cdot n$ parties.
}
\begin{proposition}
[sub-quadratic broadcast facing a constant fraction of static corruptions]\label{thm:ub:static}
\ThmUBStatic
\end{proposition}

\medskip

Perhaps more interestingly, in the regime of $n-o(n)$ static corruptions, we demonstrate a feasibility tradeoff between \emph{resiliency} and \emph{communication} that nearly tightly complements the above upper bound. We show that resilience in the face of only $\epsilon(n)\cdot n$ honest parties, for $\epsilon(n) \in o(1)$, demands message complexity scaling as $\Omega(n/\epsilon(n))$.
Note that a lower bound on message complexity is stronger than for communication complexity, directly implying the latter.
Our lower bound holds for randomized protocols, given any cryptographic assumption and any setup information that is generated by an external trusted party and given to the parties before the beginning of the protocol, including the assumptions of the above upper bound.

\def\ThmLBStatic
{
Let $\epsilon(n)\in o(1)$.
If there exists a broadcast protocol that is secure against $(1-\epsilon(n))\cdot n$ static corruptions, then the message complexity of the protocol is $\Omega(n\cdot \frac1{\epsilon(n)})$.
}
\begin{theorem}[communication lower bound for static corruptions]\label{thm:lb:static}
\ThmLBStatic
\end{theorem}

For example, for $n-n/\log^d(n)$ corruptions with a constant $d \ge 1$ (i.e., $\epsilon(n)=\log^{-d}(n)$), the message complexity must be $\Omega(n\cdot \log^d(n))$. For $n-\sqrt{n}$ corruptions (i.e., $\epsilon(n)=1/\sqrt{n}$), the message complexity must be $\Omega(n\cdot \sqrt{n})$.
And for $n-c$ corruptions with a constant $c>1$ (i.e., $\epsilon(n)=c/n$), the message complexity must be $\Omega(n^2)$, in particular meaning that sub-quadratic communication is impossible in this regime.

As noted, \cref{thm:lb:static} holds for any cryptographic and setup assumptions. This captures, in particular, PKI-style setup (such as the VRF-based PKI of Chan et al.\ \cite{CPS20}) in which the trusted party samples a private/public key-pair for each party and gives each party its private key together with the vector of all public keys. It additionally extends to even stronger, more involved setup assumptions for generating \emph{correlated randomness} beyond a product distribution, \eg setup for threshold signatures where parties' secret values are nontrivially correlated across one another.

\paragraph{Weakly adaptive adversary.}
The lower bound of \cref{thm:lb:static} carries over directly to the setting of weakly adaptive adversaries. Shifting back to the regime of a constant fraction of corruptions, one may naturally ask whether a protocol such as that from \cref{thm:ub:static} can also exist within this regime.

Unfortunately, given a few minutes thought one sees that a balanced protocol with polylogarithmic per-party communication as demonstrated by \cref{thm:ub:static} cannot translate to the weakly adaptive setting. The reason is that if the \emph{sender} party speaks to at most $t$ other parties, then the adaptive adversary can simply corrupt each receiving party and drop the message, thus blocking any information of the sender's input from reaching honest parties.

However, this attack applies only to the unique sender party. Indeed, non-sender parties contribute no input to the protocol to be blocked; and, without the ability to perform ``after-the-fact'' message removal, a weakly adaptive adversary cannot prevent communication from being \emph{received} by a party without a very large number of corruptions.

We therefore consider the locality of \emph{non-sender} parties, and ask whether sublinear locality is achievable. Our third result answers this question in the negative. That is, we show an efficient adversary who can force any party of its choosing to communicate with a large number of neighbors. Note that this in particular lower bounds the per-party communication complexity of non-sender parties.

\def\ThmLBAdaptive
{
Let $0<\locality < (n-1)/2$ and let $\pi$ be an $n$-party broadcast protocol secure against $t=n/2 + \locality$ adaptive corruptions.
Then, for any non-sender party $\Ps$ there exists a PPT adversary that can force the locality of $\Ps$ to be larger than $\locality$, except for negligible probability.
}
\begin{theorem}[non-sender locality facing adaptive corruptions]\label{thm:lb:adaptive}
\ThmLBAdaptive
\end{theorem}

For example, for $k\in\Theta(n)$, \eg a constant fraction $t=0.51 \cdot n$ of corruptions, the locality of $\Ps$ must be $\Theta(n)$, thus forming a separation from \cref{thm:ub:static} for the locality of non-sender parties. Similarly to \cref{thm:lb:static}, this bound holds in the presence of any correlated-randomness setup and for any cryptographic assumptions.

We remark that our bound further indicates a design requirement for any protocol attempting to achieve sub-quadratic $o(n^2)$ communication complexity within this setting. To obtain $o(n^2)$ communication, it must of course be that nearly all parties have sublinear communication locality.
Our result shows that any such protocol must include instructions causing a party to send out messages to a linear number of other parties upon determining that it is under attack.

\paragraph{Summary.}
For completeness, \cref{tbl:intro:summary} summarizes our results alongside prior work.

\input{tbl_summary}

\subsection{Technical Overview}\label{sec:intro:tech}

The proof of \cref{thm:ub:static} follows almost immediately from \cite{CPS20} and \cite{TLP22}. We therefore focus on our lower bounds.

\paragraph{Communication lower bound for static corruptions.}
The high-level idea of the attack underlying \cref{thm:lb:static} is to split all parties except for the sender $\sender$ into two equal-size subsets, $\calA$ and $\calB$, randomly choose a set $\calS$ of size $\epsilon(n)-1$ parties in $\calA$ and a party $\Ps\in\calB$, and corrupt all parties but $\calS\cup\sset{\Ps}$ (as illustrated in \cref{lb-figure}).
The adversary proceeds by running two independent executions of the protocol. In the first, the sender runs an execution on input $0$ towards $\calA$, and all corrupted parties in $(\sset{\sender}\cup\calA)\setminus\calS$ ignore all messages from parties in $\calB$ (pretending they all crashed). In the second, the sender runs an execution on input~$1$ towards $\calB$, and all corrupted parties in $(\sset{\sender}\cup\calB)\setminus\sset{\Ps}$ ignore all messages from parties in $\calA$.

\input{diag_lb_static}

As long as the honest parties in $\calS$ and the honest party $\Ps$ do not communicate, the adversary will make them output different values. This holds because, conditioned on no communication between $\calS$ and $\Ps$, the view of honest parties in $\calS$ is indistinguishable from a setting where the adversary crashes all parties in $\calB$ and an honest sender has input~$0$; in this case, all parties in $\calA$ (and in particular in $\calS$) must output $0$. Similarly, conditioned on no communication between $\calS$ and $\Ps$, the view of $\Ps$ is indistinguishable from a setting where the adversary crashes all parties in $\calA$ and an honest sender has input $1$; in this case, all parties in $\calB$ (and in particular $\Ps$) must output $1$.

The challenge now is to argue that the honest parties in $\calS$ and the honest party $\Ps$ do not communicate with noticeable probability. Note that this does not follow trivially from the overall low communication complexity, as the communication patterns unfold as a function of the adversarial behavior, which in particular depends on the choice of $\calS$ and $\Ps$. The argument instead follows from a series a delicate steps that compare the view of parties in this execution with other adversarial strategies.

The underlying trick is to analyze the event of communication between $\calS$ and $\Ps$ by splitting into two sub-cases: when $\calS$ speaks to $\Ps$ \emph{before} receiving any message from $\Ps$, and when $\Ps$ speaks to $\calS$ \emph{before} receiving any message from $\calS$. (Note, these events are not disjoint.) The important observation is that before any communication is received by the other side, then each side's view in the attack is identically distributed as in a hypothetical execution in which the corresponding set $\calA$ or $\calB$ crashes from the start.
Since these simple crash adversarial strategies are indeed independent of $\calS$ and $\Ps$, then we can easily analyze and upper bound the probability of $\calS$ and $\Ps$ communicating within their hypothetical executions. To finalize the argument, we carry this analysis over to show that with noticeable probability $\Ps$ does not communicate with~$\calS$ in an actual execution with the original adversary.

\paragraph{Locality lower bound for weakly adaptive corruptions.}

We proceed to consider the setting of \emph{weakly adaptive} corruptions. As mentioned above, in the adaptive setting it is easy to see that the sender must communicate with many parties, since otherwise the adversary may crash every party that the sender communicates with; therefore, the challenging part is to focus on non-sender parties. Further, when considering strong adaptive adversaries that can perform after-the-fact message removal by corrupting the sender, every honest party must communicate with a linear number of parties \cite{ACDNPRS19}. In our setting, we do not consider such capabilities of the adversary. In particular, once the adversary learns that an honest party has chosen to send a message, this message cannot be removed or changed.

Unlike our previous lower bound which assumed $n-o(n)$ corruptions, here we consider $n/2+k$ corruptions for $k\in O(n)$, so we cannot prevent sets of honest parties from communicating with each other.
Our approach, instead, is to keep the targeted party $\Ps$ confused about the output of other honest parties.

More concretely, our adversarial strategy splits all parties but the sender and $\Ps$ into disjoint equal-size sets $\calS_0$ and $\calS_1$ of parties, samples a random bit $b$ and corrupts the sender party and the parties in $\calS_{1-b}$. The adversary communicates with $\calS_0$ as if the sender's input is~$0$ and all parties in $\calS_1$ have crashed, and at the same time plays towards $\calS_1$ as if the sender's input is $1$ and all parties in $\calS_0$ have crashed (as illustrated in \cref{adaptive-lb-figure}).
Although the adversary cannot prevent honest parties from $\calS_b$ from sending messages to the targeted party $\Ps$, it can corrupt every party that \emph{receives} a message from $\Ps$. The effect of this attack is that, although $\Ps$ can tell that the sender is cheating, $\Ps$  cannot know whether parties in $\calS_0$ or parties in $\calS_1$ are honest.  And, moreover, $\Ps$ cannot know whether the remaining honest parties \emph{know} that the sender is cheating or if they believe that the sender is honest and other parties crashed---in which case they must output a bit (either $0$ if $\calS_0$ are honest or $1$ if $\calS_1$ are honest). To overcome this attack, $\Ps$ must communicate with sufficiently many parties such that the adversary's corruption budget will run out, \ie with output locality at least $k$.

\input{diag_lb_adaptive}


\subsection{Further Related Work}

Since the classical results from the '80s, a significant line of work has been devoted to understanding the complexity of broadcast protocols.\footnote{In this work we consider broadcast protocols that achieve the usual properties of \emph{termination}, \emph{agreement}, and \emph{validity}. We note that stronger notions of broadcast have been considered in the literature, \eg in the adaptive setting, the works of \cite{HZ10,GKKZ11,CGZ23} study \emph{corruption fairness} ensuring that once any receiver learns the sender's input, the adversary cannot corrupt the sender and change its message. As our main technical contributions are lower bounds, focusing on weaker requirements yields stronger results.}

\paragraph{Communication complexity.}

In the honest-majority regime, we know of several protocols, deterministic \cite{BGP92, CW92, MR21} or randomized \cite{Micali17, FLL21}, that match the known lower bounds \cite{DR85, ACDNPRS19} for strongly adaptive adversaries. When considering static, or weakly adaptive security, a fruitful line of works achieved sub-quadratic communication, with information-theoretic security \cite{KSSV06,KS09,KS11,BGH13} or with computational security \cite{CM19,ACDNPRS19,CKS20,BKLL20,BCG21} .

In the dishonest-majority regime, the most communication-efficient broadcast constructions are based on the protocol of Dolev and Strong~\cite{DS83}. This protocol is secure facing any number of strongly adaptive corruptions and the communication complexity is $O(n^3)$. When considering weakly adaptive corruptions, Chan et al.\ \cite{CPS20} used cryptography and trusted setup to dynamically elect a small, polylog-size committee in each round and improved the communication to $\tilde{O}(n^2)$. In the static-corruption setting, Tsimos et al.\ \cite{TLP22} achieved $\tilde{O}(n^2)$ communication by running the protocol of \cite{DS83} over a ``gossiping network'' \cite{DGHILSSST87,KSSV00}. This work further achieved \emph{amortized} sub-quadratic communication facing weakly adaptive corruptions when all parties broadcast in parallel.

A line of works focused on achieving \emph{balanced} protocols, where all parties incur the same work in terms of communication complexity \cite{KS11,BCG21,ADDRVXZ22}. The work of \cite{BCG21} also showed lower bounds on the necessary setup and cryptographic assumptions to achieve balanced protocols when extending almost-everywhere agreement to full agreement.\footnote{\emph{Almost-everywhere agreement} \cite{DPPU88} is a relaxed problem in which all but an $o(1)$ fraction of the parties must reach agreement. For this relaxation, King et al.\ \cite{KSSV06} showed an efficient protocol, with poly-logarithmic locality, communication, and rounds. This protocol serves as a stepping stone to several sub-quadratic Byzantine agreement protocols, by extending almost-everywhere agreement to \emph{full agreement} \cite{KSSV06,KS09,KS11,BGH13,BCG21}.}
Message dissemination protocols \cite{DXR21,LMT22} have also been proven useful for constructing balanced protocols.

The work in \cite{HKK08} showed that without trusted setup assumptions, at least one party must send $\Omega(n^{1/3})$ messages, in the \emph{static filtering} model, where each party must decide which set of parties it will accept messages from in each round before the rounds begins. We remark that our lower bounds hold also given trusted setup, and in the dynamic-filtering model (in which sub-quadratic upper bounds have been achieved).

\paragraph{Connectivity.}
Obtaining communication-efficient protocols inherently relies on using a strict subgraph of the communication network.
Early works \cite{Dolev82,FLM86} showed that \emph{deterministic} broadcast is possible in an incomplete graph only if the graph is $(t+1)$-connected. The influential work of King et al.\ \cite{KSSV06} laid a path not only for randomized Byzantine agreement with sub-quadratic communication, but also for protocols that run over a partial graph \cite{KS09,KS11,BGH13,BCG21}. The graphs induced by those protocols yield expander graphs, and the work of \cite{BCDH18} showed that in the strongly adaptive setting and facing a linear number of corruptions, no protocol for all-to-all broadcast in the plain model (without PKI setup) can maintain a non-expanding communication graph against all adversarial strategies. Further, feasibility of broadcast with a non-expander communication graph, admitting a sub-linear cut, was demonstrated in weaker settings \cite{BCDH18}.

\paragraph{Round complexity.}
In terms of round complexity, when considering deterministic protocols, $t+1$ rounds are known to be sufficient \cite{PSL80,DS83,GM93} and necessary \cite{FL82,DS83}.
Ben-Or \cite{BenOr83} and Rabin \cite{Rabin83}, showed that this lower bound can be overcome using randomization.
In the case of fixed-round protocols, the works of \cite{FM97,Micali17, FLL21} showed protocols achieving $2^{-r}$ error within $O(r)$ rounds. On the other hand, Karlin and Yao \cite{KY84} and Chor, Merritt and Shmoys \cite{CMS89} showed that any $r$-round protocol incurs an error probability of $r^{-r}$ when the number of corruptions is linear, a bound that has recently been matched by Ghinea, Goyal and Liu-Zhang \cite{GGL22}.
For protocols with probabilistic termination, randomized broadcast with expected-constant number of rounds was achieved in the honest-majority setting \cite{FM97,FG03,KK06}, even under composition \cite{CCGZ19,CCGZ21}. It was further shown that two rounds are unlikely to suffice for reaching agreement, even with weak guarantees, as long as $t > n/4$ \cite{CHM22} (as opposed to \emph{three} rounds \cite{Micali17}).
In the dishonest-majority setting, there are sublinear-round broadcast protocols \cite{GKKO07,FN09,CPS20,WXDS20}, and even expected-constant-round protocols \cite{WXSD20,SLMNPT23}. These results match the lower bound of $\Omega(n/(n-t))$ rounds (allowing up to constant failure probability) \cite{GKKO07}.

\subsection*{Outline of Paper}
The paper is organized as follows. In \cref{sec:Preliminaries}, preliminary content including  notations, security, and network model is introduced. In \cref{sec:lb:static}, we present  the message-complexity lower bound for static corruptions. In \cref{sec:lb:adaptive}, we present the locality lower bound for weakly adaptive corruptions. Finally, in \cref{sec:floodbroadcast}, we describe a statically secure broadcast protocol with sub-quadratic communication and poly-logarithmic locality.

%% file: tbl_summary.tex
\renewcommand{\arraystretch}{1.3}

\begin{table}[!htb]
\begin{center}
\begin{footnotesize}
\begin{tabular}{@{}l|ll|lll@{}}
    \toprule
    ~ &
    \textbf{setup} &
    \textbf{corruptions} &
    \textbf{total com.} &
    \Centerstack[l]{\textbf{locality}\\\textbf{(non-sender)}} &
    \textbf{ref.} \\

    \midrule
    \midrule
    \Centerstack[l]{\textbf{strongly}\\\textbf{adaptive}} &
    \Centerstack[l]{bare pki\\any} &
    \Centerstack[l]{$t<n$\\$t=\Theta(n)$} &
    \Centerstack[l]{$O(n^3)$\\$\Omega(n^2)$} &
    \Centerstack[l]{$n$\\$\Omega(n)$} &
    \Centerstack[l]{\cite{DS83}\\\cite{ACDNPRS19}} \\

    \midrule
    \Centerstack[l]{\textbf{weakly}\\\textbf{adaptive}} &
    \Centerstack[l]{trusted pki\\ any} &
    \Centerstack[l]{$t=\Theta(n)$\\ $t=\Theta(n)$} &
    \Centerstack[l]{$\tilde O(n^2)$\\ ~} &
    \Centerstack[l]{$O(n)$\\ $\Omega(n)$} &
    \Centerstack[l]{\cite{CPS20}\\ Thm.~\ref{thm:lb:adaptive}} \\

    \midrule
    ~ &
    \Centerstack[l]{any\\ (deterministic)} &
    $t=\Theta(n)$ &
    $\Omega(n^2)$ &
    $\Omega(n)$ &
    \cite{DR85} \\

    ~ &
    bare pki &
    $t=\Theta(n)$ &
    $\tilde O(n^2)$ &
    $\tilde{O}(1)$ & 
    \cite{TLP22} \\

    \textbf{static} &
    trusted pki &
    $t=\Theta(n)$ &
    $\tilde O(n)$ &
    $\tilde{O}(1)$ & 
    Prop.~\ref{thm:ub:static} \\

    ~ &
    any &
    \Centerstack[l]{
    $t=(1-\epsilon(n))\cdot n$, $\epsilon(n)\in o(1)$\\
    \eg $t=n-\frac{n}{\polylog(n)}$\\
    \eg $t=n-\sqrt{n}$\\
    \eg $t=n-O(1)$
    } &
    \Centerstack[l]{
    $\Omega(n\cdot \frac1{\epsilon(n)})$\\
    $\Omega(n\cdot \polylog(n))$\\
    $\Omega(n\cdot \sqrt n)$\\
    $\Omega(n^2)$
    } &
    ~ &
    Thm.~\ref{thm:lb:static} \\

    \bottomrule
\end{tabular}\ \\[1ex]

\caption{\footnotesize{Communication requirements of dishonest-majority (synchronous) broadcast. We consider the standard, property-based definition of broadcast (see Definition~\ref{def:property-bc}). For each type of adversary (strongly adaptive, weakly adaptive, and static), we consider the state-of-the-art protocols and lower bounds in terms of setup, number of corruptions, total communication and (non-sender) locality. For setup we distinguish bare PKI, where each party locally generates its key pair, as opposed to trusted PKI, where all keys are generated by a trusted dealer. Reference \cite{DR85} is only for deterministic protocols.}}
\label{tbl:intro:summary}
\end{footnotesize}
\end{center}
\vspace*{-5ex}
\end{table}

%% file: diag_lb_static.tex
\begin{figure}[!htb]
\centering
\tikzset{every picture/.style={line width=0.75pt}} 

\begin{tikzpicture}[x=0.75pt,y=0.75pt,yscale=-1,xscale=1]

\draw   (140,50) -- (470,50) -- (470,260) -- (140,260) -- cycle ;

\draw  [fill={rgb, 255:red, 217; green, 217; blue, 217 }  ,fill opacity=1 ] (320,135.4) .. controls (320,121.37) and (331.37,110) .. (345.4,110) -- (421.6,110) .. controls (435.63,110) and (447,121.37) .. (447,135.4) -- (447,224.6) .. controls (447,238.63) and (435.63,250) .. (421.6,250) -- (345.4,250) .. controls (331.37,250) and (320,238.63) .. (320,224.6) -- cycle ;
\draw  [fill={rgb, 255:red, 217; green, 217; blue, 217 }  ,fill opacity=1 ] (163,135.4) .. controls (163,121.37) and (174.37,110) .. (188.4,110) -- (264.6,110) .. controls (278.63,110) and (290,121.37) .. (290,135.4) -- (290,224.6) .. controls (290,238.63) and (278.63,250) .. (264.6,250) -- (188.4,250) .. controls (174.37,250) and (163,238.63) .. (163,224.6) -- cycle ;
\draw  [fill={rgb, 255:red, 255; green, 255; blue, 255 }  ,fill opacity=1 ] (390,175) .. controls (390,172.24) and (392.24,170) .. (395,170) .. controls (397.76,170) and (400,172.24) .. (400,175) .. controls (400,177.76) and (397.76,180) .. (395,180) .. controls (392.24,180) and (390,177.76) .. (390,175) -- cycle ;
\draw  [fill={rgb, 255:red, 255; green, 255; blue, 255 }  ,fill opacity=1 ] (183,182) .. controls (183,175.37) and (188.37,170) .. (195,170) -- (261,170) .. controls (267.63,170) and (273,175.37) .. (273,182) -- (273,218) .. controls (273,224.63) and (267.63,230) .. (261,230) -- (195,230) .. controls (188.37,230) and (183,224.63) .. (183,218) -- cycle ;
\draw  [fill={rgb, 255:red, 217; green, 217; blue, 217 }  ,fill opacity=1 ] (300,75) .. controls (300,72.24) and (302.24,70) .. (305,70) .. controls (307.76,70) and (310,72.24) .. (310,75) .. controls (310,77.76) and (307.76,80) .. (305,80) .. controls (302.24,80) and (300,77.76) .. (300,75) -- cycle ;

\draw (311,62.4) node [anchor=north west][inner sep=0.75pt]  [font=\large]  {$\Party_{s}$};
\draw (171,118.4) node [anchor=north west][inner sep=0.75pt]  [font=\large]  {$\mathcal{A}$};
\draw (412.5,164.5) node  [font=\large]  {$\Ps$};
\draw (194,178.4) node [anchor=north west][inner sep=0.75pt]  [font=\large]  {$\mathcal{S}$};
\draw (331,118.4) node [anchor=north west][inner sep=0.75pt]  [font=\large]  {$\mathcal{B}$};

\draw (281,271) node [anchor=north west][inner sep=0.75pt]   [align=left] {Corrupt};
\draw (391,271) node [anchor=north west][inner sep=0.75pt]   [align=left] {Honest};
\draw (197,271) node [anchor=north west][inner sep=0.75pt]   [align=left] {Key:};
\draw  [fill={rgb, 255:red, 217; green, 217; blue, 217 }  ,fill opacity=1 ] (250,269) -- (270,269) -- (270,289) -- (250,289) -- cycle;
\draw  [fill={rgb, 255:red, 255; green, 255; blue, 255 }  ,fill opacity=1 ] (360,270) -- (380,270) -- (380,290) -- (360,290) -- cycle;
\end{tikzpicture}

\caption{The partition of the parties used in the proof of \cref{thm:lb:static}. All the parties except for the sender $\Party_s$ are partitioned into two equal-size sets, $\calA$ and $\calB$. A subset $\calS\subseteq \calA$ of size $\epsilon(n)\cdot n-1$ and a party $\Ps\in\calB$ are uniformly sampled, and the adversary statically corrupts all parties but $\calS\cup\sset{\Ps}$. The goal of the attack is to ensure that $\calS$ and $\Ps$ do not communicate, thus forcing them to output different bits.
}
\label{lb-figure}
\end{figure}

%% file: diag_lb_adaptive.tex
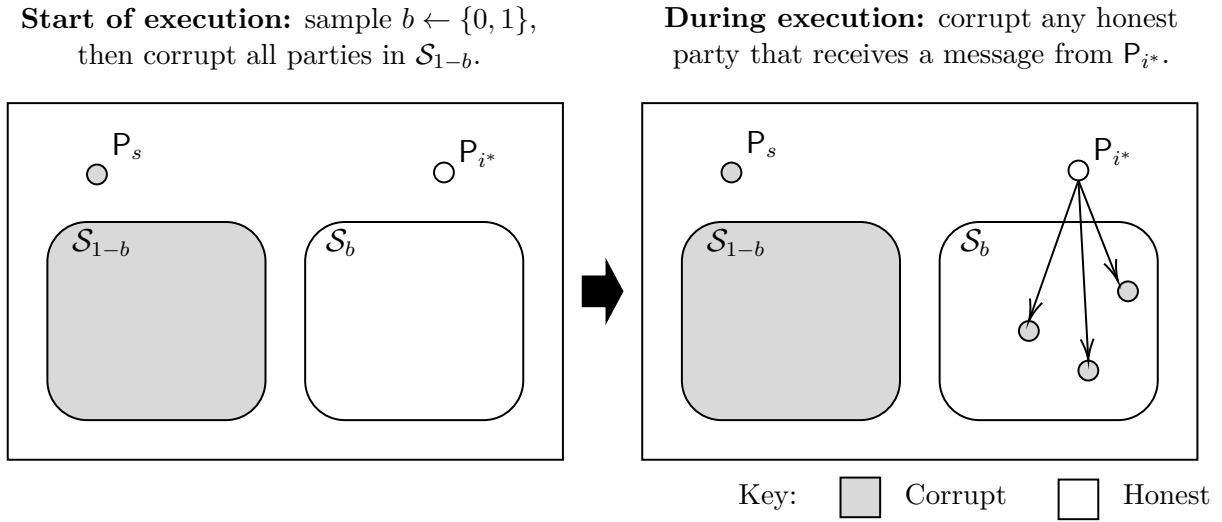
\begin{figure}[!htb]
\centering

\tikzset{every picture/.style={line width=0.75pt}} 

\begin{tikzpicture}[x=0.75pt,y=0.75pt,yscale=-1,xscale=1]

\draw  [fill={rgb, 255:red, 217; green, 217; blue, 217 }  ,fill opacity=1 ] (450,269) -- (470,269) -- (470,289) -- (450,289) -- cycle ;
\draw  [fill={rgb, 255:red, 255; green, 255; blue, 255 }  ,fill opacity=1 ] (560,270) -- (580,270) -- (580,290) -- (560,290) -- cycle ;
\draw  [fill={rgb, 255:red, 217; green, 217; blue, 217 }  ,fill opacity=1 ] (390,115) .. controls (390,112.24) and (392.24,110) .. (395,110) .. controls (397.76,110) and (400,112.24) .. (400,115) .. controls (400,117.76) and (397.76,120) .. (395,120) .. controls (392.24,120) and (390,117.76) .. (390,115) -- cycle ;
\draw  [fill={rgb, 255:red, 255; green, 255; blue, 255 }  ,fill opacity=1 ] (565,114) .. controls (565,111.24) and (567.24,109) .. (570,109) .. controls (572.76,109) and (575,111.24) .. (575,114) .. controls (575,116.76) and (572.76,119) .. (570,119) .. controls (567.24,119) and (565,116.76) .. (565,114) -- cycle ;
\draw   (350,80) -- (630,80) -- (630,260) -- (350,260) -- cycle ;
\draw  [fill={rgb, 255:red, 255; green, 255; blue, 255 }  ,fill opacity=1 ] (500,160) .. controls (500,148.95) and (508.95,140) .. (520,140) -- (590,140) .. controls (601.05,140) and (610,148.95) .. (610,160) -- (610,220) .. controls (610,231.05) and (601.05,240) .. (590,240) -- (520,240) .. controls (508.95,240) and (500,231.05) .. (500,220) -- cycle ;
\draw  [fill={rgb, 255:red, 217; green, 217; blue, 217 }  ,fill opacity=1 ] (540,195) .. controls (540,192.24) and (542.24,190) .. (545,190) .. controls (547.76,190) and (550,192.24) .. (550,195) .. controls (550,197.76) and (547.76,200) .. (545,200) .. controls (542.24,200) and (540,197.76) .. (540,195) -- cycle ;
\draw  [fill={rgb, 255:red, 217; green, 217; blue, 217 }  ,fill opacity=1 ] (570,215) .. controls (570,212.24) and (572.24,210) .. (575,210) .. controls (577.76,210) and (580,212.24) .. (580,215) .. controls (580,217.76) and (577.76,220) .. (575,220) .. controls (572.24,220) and (570,217.76) .. (570,215) -- cycle ;
\draw  [fill={rgb, 255:red, 217; green, 217; blue, 217 }  ,fill opacity=1 ] (590,175) .. controls (590,172.24) and (592.24,170) .. (595,170) .. controls (597.76,170) and (600,172.24) .. (600,175) .. controls (600,177.76) and (597.76,180) .. (595,180) .. controls (592.24,180) and (590,177.76) .. (590,175) -- cycle ;
\draw    (570,119) -- (545.66,188.11) ;
\draw [shift={(545,190)}, rotate = 289.4] [color={rgb, 255:red, 0; green, 0; blue, 0 }  ][line width=0.75]    (10.93,-3.29) .. controls (6.95,-1.4) and (3.31,-0.3) .. (0,0) .. controls (3.31,0.3) and (6.95,1.4) .. (10.93,3.29)   ;
\draw    (570,119) -- (574.89,208) ;
\draw [shift={(575,210)}, rotate = 266.86] [color={rgb, 255:red, 0; green, 0; blue, 0 }  ][line width=0.75]    (10.93,-3.29) .. controls (6.95,-1.4) and (3.31,-0.3) .. (0,0) .. controls (3.31,0.3) and (6.95,1.4) .. (10.93,3.29)   ;
\draw    (570,119) -- (589.27,168.14) ;
\draw [shift={(590,170)}, rotate = 248.59] [color={rgb, 255:red, 0; green, 0; blue, 0 }  ][line width=0.75]    (10.93,-3.29) .. controls (6.95,-1.4) and (3.31,-0.3) .. (0,0) .. controls (3.31,0.3) and (6.95,1.4) .. (10.93,3.29)   ;
\draw  [fill={rgb, 255:red, 217; green, 217; blue, 217 }  ,fill opacity=1 ] (370,160) .. controls (370,148.95) and (378.95,140) .. (390,140) -- (460,140) .. controls (471.05,140) and (480,148.95) .. (480,160) -- (480,220) .. controls (480,231.05) and (471.05,240) .. (460,240) -- (390,240) .. controls (378.95,240) and (370,231.05) .. (370,220) -- cycle ;
\draw   (30,80) -- (310,80) -- (310,260) -- (30,260) -- cycle ;
\draw  [fill={rgb, 255:red, 255; green, 255; blue, 255 }  ,fill opacity=1 ] (180,160) .. controls (180,148.95) and (188.95,140) .. (200,140) -- (270,140) .. controls (281.05,140) and (290,148.95) .. (290,160) -- (290,220) .. controls (290,231.05) and (281.05,240) .. (270,240) -- (200,240) .. controls (188.95,240) and (180,231.05) .. (180,220) -- cycle ;
\draw  [fill={rgb, 255:red, 217; green, 217; blue, 217 }  ,fill opacity=1 ] (50,160) .. controls (50,148.95) and (58.95,140) .. (70,140) -- (140,140) .. controls (151.05,140) and (160,148.95) .. (160,160) -- (160,220) .. controls (160,231.05) and (151.05,240) .. (140,240) -- (70,240) .. controls (58.95,240) and (50,231.05) .. (50,220) -- cycle ;
\draw  [fill={rgb, 255:red, 217; green, 217; blue, 217 }  ,fill opacity=1 ] (70,116) .. controls (70,113.24) and (72.24,111) .. (75,111) .. controls (77.76,111) and (80,113.24) .. (80,116) .. controls (80,118.76) and (77.76,121) .. (75,121) .. controls (72.24,121) and (70,118.76) .. (70,116) -- cycle ;
\draw  [fill={rgb, 255:red, 255; green, 255; blue, 255 }  ,fill opacity=1 ] (245,115) .. controls (245,112.24) and (247.24,110) .. (250,110) .. controls (252.76,110) and (255,112.24) .. (255,115) .. controls (255,117.76) and (252.76,120) .. (250,120) .. controls (247.24,120) and (245,117.76) .. (245,115) -- cycle ;
\draw  [fill={rgb, 255:red, 0; green, 0; blue, 0 }  ,fill opacity=1 ] (320,167.5) -- (332,167.5) -- (332,160) -- (340,175) -- (332,190) -- (332,182.5) -- (320,182.5) -- cycle ;

\draw (167.05,47.29) node   [align=left] {\begin{minipage}[lt]{200pt}\setlength\topsep{0pt}
\begin{center}
\textbf{Start of execution:} sample $\displaystyle b\leftarrow \{0,1\}$, then corrupt all parties in $\displaystyle \calS_{1-b}$.
\end{center}

\end{minipage}};
\draw (481,271) node [anchor=north west][inner sep=0.75pt]   [align=left] {Corrupt};
\draw (591,271) node [anchor=north west][inner sep=0.75pt]   [align=left] {Honest};
\draw (397,271) node [anchor=north west][inner sep=0.75pt]   [align=left] {Key:};
\draw (490.79,47.29) node   [align=left] {\begin{minipage}[lt]{202.36pt}\setlength\topsep{0pt}
\begin{center}
\textbf{During execution:} corrupt any honest party that receives a message from $\displaystyle \Ps$.
\end{center}

\end{minipage}};
\draw (401,92.4) node [anchor=north west][inner sep=0.75pt]  [font=\large]  {$\Party_{s}$};
\draw (381,142.4) node [anchor=north west][inner sep=0.75pt]  [font=\large]  {$\mathcal{S}_{1-b}$};
\draw (587.5,104.5) node  [font=\large]  {$\Ps$};
\draw (509,142.4) node [anchor=north west][inner sep=0.75pt]  [font=\large]  {$\mathcal{S}_{b}$};
\draw (61,142.4) node [anchor=north west][inner sep=0.75pt]  [font=\large]  {$\mathcal{S}_{1-b}$};
\draw (189,142.4) node [anchor=north west][inner sep=0.75pt]  [font=\large]  {$\mathcal{S}_{b}$};
\draw (81,93.4) node [anchor=north west][inner sep=0.75pt]  [font=\large]  {$\Party_{s}$};
\draw (267.5,105.5) node  [font=\large]  {$\Ps$};
\end{tikzpicture}

\caption{The partition of the parties used in the proof of Theorem~\ref{thm:lb:adaptive}. At the start of the execution, the adversary partitions all parties except for the sender $\Party_s$ and targeted party $\Ps$ into equal-size sets $\calS_0$ and $\calS_1$, and corrupts the parties in $\calS_{1-b}$ (for a random bit $b$). During the execution, whenever an honest party $\Party_j\in\calS_b$ receives a message from the targeted party $\Ps$, the adversary corrupts $\Party_j$.
}
\label{adaptive-lb-figure}
\end{figure}

%% file: Preliminaries.tex
\section{Preliminaries}\label{sec:Preliminaries}

In this section, we present the security model and preliminary definitions.

\paragraph{Notations.}\label{sec:notations}
We use calligraphic letters to denote sets or distributions (\eg $\calS$), uppercase for random variables (\eg $R$), lowercase for values (\eg $r$), and sans-serif (\eg \textsf{A}) for algorithms (\ie Turing machines).
For $n\in\mathbb{N}$, let $[n]=\sset{1,\ldots,n}$.
Let $\poly$ denote the set all positive polynomials and let PPT denote a probabilistic (interactive) Turing machines that runs in \emph{strictly} polynomial time.
We denote by $\secParam$ the security parameter.
A function $\nu \colon \mathbb{N} \mapsto [0,1]$ is \emph{negligible}, denoted $\nu(\secParam) = \negl(\secParam)$, if $\nu(\secParam)<1/p(\secParam)$ for every $p\in\poly$ and sufficiently large $\secParam$. Moreover, we say that $\nu \colon \mathbb{N} \mapsto [0,1]$ is \emph{noticeable} if $\nu(\secParam)\ge 1/p(\secParam)$ for some $p\in\poly$ and sufficiently large $\secParam$.
When using the $\tilde{O}(n)$ notation, polynomial factors in $\log(n)$ and the security parameter $\secParam$ are omitted.

\paragraph{Protocols.}
All protocols considered in this paper are PPT (probabilistic polynomial time): the running time of every party is polynomial in the (common) security parameter $\secParam$, given as a unary string. For simplicity, we consider Boolean-input Boolean-output protocols, where apart from the common security parameter, a designated sender $\Party_s$ has a single input bit, and each of the honest parties outputs a single bit. We note that our protocols can be used for broadcasting longer strings, with an additional dependency of the communication complexity on the input-string length.

As our main results are lower bounds, we consider protocols in the \emph{correlated randomness} model; that is, prior to the beginning of the protocol $\pi$ a trusted dealer samples values $(r_1,\ldots,r_n)\gets\calD_\pi$ from an efficiently sampleable known distribution $\calD_\pi$ and gives the value $r_i$ to party $\Party_i$. This model captures, for example, a trusted PKI setup for digital signatures and verifiable random functions (VRFs), where the dealer samples a public/private keys for each party and hands to each $\Party_i$ its secret key and a vector of all public keys; this is the setup needed for our upper bound result in \cref{sec:ub}.
The model further captures more involved distributions, such as setup for \emph{threshold signatures}, information-theoretic PKI \cite{PW92}, pairwise correlations for \emph{oblivious transfer} \cite{Beaver95}, and more.

We define the \emph{view} of a party $\Party_i$ as its setup information $r_i$, its random coins, possibly its input (in case $\Party_i$ is the sender), and its set of all messages received during the protocol.

\paragraph{Communication model.}

The communication model that we consider is \emph{synchronous}, meaning that protocols proceed in rounds. In each round every party can send a message to every other party over an authenticated channel, where the adversary can see the content of all transmitted messages, but cannot drop/inject messages. We emphasize that our lower bounds hold also in the private-channel setting which can be established over authenticated channels using public-key encryption and a PKI setup; our protocol construction only requires authenticated channels. It is guaranteed that every message sent in a round will arrive at its destination by the end of that round. The adversary is \emph{rushing} in the sense that it can use the messages received by corrupted parties from honest parties in a given round to determine the corrupted parties' messages for that round.

\paragraph{Adversary model.}

The adversary runs in probabilistic polynomial time and may corrupt a subset of the parties and instruct them to behave in an arbitrary (malicious) manner.
Some of our results (the lower bound in \cref{sec:lb:static} and the protocol in \cref{sec:ub})
consider a static adversary that chooses which parties to corrupt \emph{before} the beginning of the protocol, \ie \emph{before} the setup information is revealed to the parties. Note that this strengthens the lower bound, but provides a weaker feasibility result. Our second lower bound (\cref{sec:lb:adaptive}) considers an adaptive adversary that can choose which parties to corrupt during the course of the protocol, based on information it dynamically learns. We consider the \emph{atomic-multisend model} (also referred to as a weakly adaptive adversary), meaning that once a party $\Party_i$ starts sending messages in a given round, it cannot be corrupted until it completes sending all messages for that round, and every message sent by $\Party_i$ is delivered to its destination. This is weaker than the standard model for adaptive corruptions \cite{Thesis:Feldman88,CFGN96,Canetti01} (also referred to as a strongly rushing adversary), which enables the adversary to corrupt a party at any point during the protocol and drop/change messages that were not delivered yet. Again, we note that the weaker model we consider yields a stronger lower bound. Further, in the stronger model, a result by Abraham et al.\ \cite{ACDNPRS19} rules out sub-quadratic protocols with linear resiliency, even in the honest-majority setting.

\paragraph{Broadcast.}\label{sec:bc}
We consider the standard, property-based definition of broadcast.
\begin{definition}[Broadcast protocol]\label{def:property-bc}
An $n$-party protocol $\pi$, where a distinguished sender $\Party_s$ holds an initial input message $x\in\zo$, is a \textsf{broadcast protocol} secure against $t$ corruptions, if the following conditions are satisfied for any PPT adversary that corrupts up to $t$ parties:
\begin{itemize}
\item
\textbf{Termination:}
There exists an a-priori-known round $R$ such that the protocol is guaranteed to complete within $R$ rounds (\ie every so-far honest party produces an output value).
\item
\textbf{Agreement:}
For every pair of parties $\Party_i$ and $\Party_j$ that are honest at the end of the protocol, if party $\Party_i$ outputs $y_i$ and party $\Party_j$ outputs $y_j$, then $y_i=y_j$ with all but negligible probability in $\secParam$.
\item
\textbf{Validity:}
If the sender is honest at the end of the protocol, then for every party $\Party_i$ that is honest at the end of the protocol, if $\Party_i$ outputs $y_i$ then $y_i=x$ with all but negligible probability in $\secParam$.
\end{itemize}
\end{definition}

\newcommand{\OutEdges}{\mathsf{OutEdges}}
The communication locality \cite{BGT13,BCDH18} of a protocol corresponds to the maximal degree of any honest party in the communication graph induced by the protocol execution. While defining the incoming communication edges to a party can be subtle (as adversarial parties may ``spam'' honest parties; see \eg a discussion in \cite{BCDH18}), out-edges of honest parties are clearly identifiable from the protocol execution.
In this paper, we will focus on this simpler notion of output-locality, and use the terminology \emph{locality} of the protocol to simply refer to this value.
Our results provide a lower bound on output locality of given protocols, which in turn directly lower bounds standard locality as in \cite{BGT13,BCDH18}.

\begin{definition}[Output locality]
An $n$-party $t$-secure broadcast protocol $\pi$ with setup distribution $\calD_\pi$ has locality $\ell$, if for every PPT adversary $\Adv$ corrupting up to $t$ parties and every sender input $x$ it holds that
\[
\pr{\OutEdges(\pi, \Adv, \calD_\pi, \secParam, x) > \ell} \le \negl(\secParam),
\]
where $\OutEdges(\pi, \Adv, \calD_\pi, \secParam, x)$ is the random variable of the maximum number of parties any honest party sends messages to, defined by running the protocol $\pi$ with the adversary $\Adv$ and setup distribution $\calD_\pi$, security parameter $\secParam$ and sender input $x$. The probability is taken over the random coins of the honest parties, the random coins of $\Adv$, and the sampling coins from the setup distribution $\calD_\pi$.
\end{definition}

%% file: lowerbound.tex
\section{Message-Complexity Lower Bound for Static Corruptions}\label{sec:lb:static}
We begin with the proof of \cref{thm:lb:static}.
The high-level idea of the lower bound is that if a protocol has $o(n^2)$ messages, then, with noticeable probability, a randomly chosen pair of parties do not communicate even under certain attacks.

\begin{theorem}[\cref{thm:lb:static}, restated]\label{thm:LB1}
\ThmLBStatic
\end{theorem}

\begin{proof}
Let $\psi(n) = \frac{1}{12\epsilon(n)}$ and let $\pi$ be a broadcast protocol with message complexity $\MC = n \cdot \psi(n)$ that is secure against $(1-\epsilon(n))\cdot n$ static corruptions.
(In fact, we will prove a stronger statement than claimed, where the message complexity of the protocol must be greater than $n \cdot \frac{1}{12\epsilon(n)}$.)
Without loss of generality, we assume that the setup information sampled before the beginning of the protocol $(r_1,\ldots,r_n)\gets\calD_\pi$ includes the random string used by each party. That is, every party $\Party_i$ generates its messages in each round as a function of $r_i$, possibly its input (if $\Party_i$ is the sender), and its incoming messages in prior rounds.
Again, without loss of generality, let $\Party_1$ denote be the sender, and split the remaining parties to two equal-size subsets $\calA$ and $\calB$ (for simplicity, assume that $n$ is odd).

Consider the adversary $\Adv_1$ that proceeds as follows:
\begin{enumerate}
\item Choose randomly a set $\calS \subseteq \calA$ of size $\epsilon(n)\cdot n - 1$ and a party $\Ps \in \calB$.
\item Corrupt all parties except for $\calS \cup \sset{\Ps}$.
\item Receive the setup information of the corrupted parties $\sset{r_i \mid \Party_i \notin \calS \cup \sset{\Ps}}$.
\item Maintain two independent executions, denoted $\exec_0$ and $\exec_1$, as follows.
\begin{itemize}
    \item In the execution $\exec_0$, the adversary runs in its head the parties in $\calA\setminus \calS$ honestly on their setup information $\sset{r_i \mid \Party_i \in\calA\setminus\calS}$ and a copy of the sender, denoted $\Party_1^0$, running on input $0$ and setup information $r_1$.

    The adversary communicates on behalf of the virtual parties in $(\calA\setminus\calS) \cup \sset{\Party_1^0}$ with the honest parties in $\calS$ according to this execution. Every corrupted party in $\calB\setminus\sset{\Ps}$ crashes in this execution, and the adversary drops every message sent by the virtual parties in $(\calA\setminus\calS) \cup \sset{\Party_1^0}$ to $\Ps$ and does not deliver any message from $\Ps$ to these~parties.
    \item In the execution $\exec_1$, the adversary runs in its head the parties in $\calB\setminus \sset{\Ps}$ honestly on their setup information $\sset{r_i \mid \Party_i \in\calB\setminus \sset{\Ps}}$ and a copy of the sender, denoted $\Party_1^1$, running on input $1$ and setup information $r_1$.

    The adversary communicates on behalf of the virtual parties in $(\calB\setminus \sset{\Ps}) \cup \sset{\Party_1^1}$ with the honest $\Ps$ according to this execution. Every corrupted party in $\calA\setminus\calS$ crashes in this execution, and the adversary drops every message sent by the virtual parties in $(\calB\setminus \sset{\Ps}) \cup \sset{\Party_1^1}$ to honest parties in $\calS$ and does not deliver any message from $\calS$ to these parties.
\end{itemize}
\end{enumerate}

We start by defining a few notations.
Consider the following random variables
\[
\inputCoins=\left(R_1,\ldots,R_n,\Srv,\IS\right),
\]
where $R_1,\ldots,R_n$ are distributed according to $\calD_\pi$, and $\Srv$ takes a value uniformly at random in the subsets of $\calA$ of size $\epsilon(n)\cdot n - 1$, and $\IS$ takes a value uniformly at random in $\calB$.
During the proof, $R_i$ represents the setup information (including private randomness) of party $\Party_i$, whereas the pair $(\Srv,\IS)$ corresponds to the random coins of the adversary $\Adv_1$ (used for choosing $\calS$ and $\Ps$).
Unless stated otherwise, all probabilities are taken over these random variables.

Let $\mainAttack$ be the random variable defined by running the protocol $\pi$ with the adversary $\Adv_1$ over $\inputCoins$.
That is, $\mainAttack$ consists of a vector of $n+1$ views: of the honest parties in $\Srv\cup\sset{\Party_\IS}$ and of the corrupted parties in $\calA\setminus \Srv$ and $\calB\setminus\sset{\Party_\IS}$, where the \ith view is denoted by $\viewmain_i$, and of two copies of the sender $\Party_1^0$ and $\Party_1^1$, denoted $\viewmain_{1\mhyphen0}$ and $\viewmain_{1\mhyphen1}$, respectively. Each view consists of the setup information $R_i$, possibly the input, and the set of received messages in each round. Specifically,
\[
\mainAttack=\left(\viewmain_{1\mhyphen0},\viewmain_{1\mhyphen1},\viewmain_2,\ldots,\viewmain_n\right).
\]

Denote by $\attackDisconnect$ the event that $\Party_\IS$ and $\Srv$ do not communicate in $\mainAttack$; that is, $\Party_\IS$ does not send any message to parties in $\Srv$ (according to $\viewmain_\IS$) and every party $\Party_J$ with $J\in \Srv$ does not send any message to $\Party_\IS$ (according to $\viewmain_J$).
We proceed to prove that the event $\attackDisconnect$ occurs with noticeable probability.
\begin{lemma}\label{lem:lbstatic_disconnect}
$\pr{\attackDisconnect} \geq  \frac13$.
\end{lemma}
\begin{proof}
Denote by $\mainStoP$ the event that a party in $\Srv$ sends a message to $\Party_\IS$ in $\mainAttack$, and $\Party_\IS$ did not send any message to any party in $\Srv$ in any prior round.
We begin by upper bounding the probability of $\mainStoP$.
\begin{claim}
$\pr{\mainStoP} \leq  \frac13$.
\end{claim}
\begin{proof}
Consider a different adversary for $\pi$, denoted $\Adv_B$, that statically corrupts all parties in $\calB$ and crashes them (all other parties including the sender are honest).
Let $\thirdAttack$ denote the random variable defined by running the protocol $\pi$ with the adversary $\Adv_B$ over $\inputCoins$, in which the honest sender's input is~$1$.
That is, $\thirdAttack$ consists of a vector of $n/2+1$ views: of the honest parties in $\calA$, where the \ith view is denoted by $\viewcrashB_i$, and the sender $\Party_1$ denoted by $\viewcrashB_1$. Each view consists of the setup information $R_i$, the input $1$ for $\Party_1$, and the set of received messages in each round. Specifically,
\[
\thirdAttack=\left(\viewcrashB_i\right)_{i\in\calA\cup\sset{1}}.
\]

Denote by $\crashBStoP$ the event that a party in $\Srv$ sends a message to $\Party_\IS$ in $\thirdAttack$ such that $\Party_\IS$ did not send any message to any party in $\Srv$ in any prior round.
Note that as long as parties in $\Srv$ do not receive a message from $\Party_\IS$ until some round $\rho$ in $\mainAttack$, their joint view is identically distributed as their joint view in $\thirdAttack$ up until round $\rho$. Therefore,
\begin{equation*}
\pr{\mainStoP}=\pr{\crashBStoP}.
\end{equation*}

Note that, by the definition of $\Adv_B$, the distribution of $\thirdAttack$, and therefore $\pr{\crashBStoP}$, is \emph{independent} of the random variables $\Srv$ and $\IS$. Hence, one can consider the mental experiment where $R_1,\ldots,R_n$ are first sampled for setting $\thirdAttack$, and later, $\Srv$ and $\IS$ are independently sampled at random. This does not affect the event $\crashBStoP$.

Recall that the message complexity of $\pi$ is $\MC=n \cdot \psi(n)$ for $\psi(n)=\frac{1}{12\epsilon(n)}$. Further, $\Srv$ is of size $\ssize{\Srv}=\epsilon(n)\cdot n - 1$ and $\ssize{\calA}=\ssize{\calB}=n/2$. Observe that the message complexity upper-bounds the number of communication edges between $\calA$ and $\calB$. Further, the probability that a party in $\Srv$ talks first to $\Party_\Is$ is upper-bounded by the probability that there exists a communication edge between $\Srv$ and $\Party_\Is$. Since $\Srv$ and $\IS$ are uniformly distributed in $\calA$ and $\calB$, respectively, we obtain that this probability is bounded by
\begin{align*}
\pr{\crashBStoP}
&\leq \MC\cdot\frac{1}{\ssize{\calB}}\cdot \frac{\ssize{\Srv}}{\ssize{\calA}} \\
&= n \cdot \psi(n) \cdot \frac{1}{n/2} \cdot \frac{\epsilon(n)\cdot n-1}{n/2} \\
&\leq n \cdot \psi(n) \cdot \frac{1}{n/2} \cdot \frac{\epsilon(n)\cdot n}{n/2} \\
&= 4 \cdot \psi(n) \cdot \epsilon(n)\\
&=\frac{4 \cdot \epsilon(n)}{12 \cdot \epsilon(n)} =\frac13.
\qedhere
\end{align*}
\end{proof}

Similarly, denote by $\mainPtoS$ the event that $\Party_\IS$ sends a message to a party in $\Srv$ in $\mainAttack$, such that no party in $\Srv$ sent a message to $\Party_\IS$ in any prior round; i.e., changing the order from $\mainStoP$. We upper bound the probability of $\mainPtoS$ in an analogous manner.
\begin{claim}
$\pr{\mainPtoS} \leq  \frac13$.
\end{claim}
\begin{proof}
Consider a different adversary for $\pi$, denoted $\Adv_A$, that statically corrupts all parties in~$\calA$ and crashes them.
Let $\secondAttack$ be a random variable defined by running the protocol $\pi$ with the adversary $\Adv_A$ over $\inputCoins$, in which the honest sender's input is $0$.
That is, $\secondAttack$ consists of a vector of $n/2+1$ views: of the honest parties in $\calB$, where the \ith view is denoted by $\viewcrashA_i$, and the sender $\Party_1$ denoted by $\viewcrashA_1$. Each view consists of the setup information $R_i$, the input $0$ for $\Party_1$, and the set of received messages in each round. Specifically,
\[
\secondAttack=\left(\viewcrashA_i\right)_{i\in\calB\cup\sset{1}}.
\]

Denote by $\crashAPtoS$ the event that $\Party_\IS$ sends a message to a party in $\Srv$ in $\secondAttack$, and no party in $\Srv$ sent a message to $\Party_\IS$ in any prior round.
As long as $\Party_\IS$ does not receive a message from parties in $\Srv$  until some round $\rho$ in $\mainAttack$, its view is identically distributed as its view in $\secondAttack$ up until round $\rho$. Therefore,
\begin{equation*}
\pr{\mainPtoS}=\pr{\crashAPtoS}.
\end{equation*}
An analogue analysis to the previous case shows that $\pr{\crashAPtoS} \leq 1/3$, as desired.
\end{proof}

Combined together, we get that
\[
\pr{\neg \attackDisconnect} = \pr{\mainStoP \cup \mainPtoS} \leq \pr{\mainStoP} + \pr{\mainPtoS} \leq\frac23.
\]
Therefore, $\pr{\attackDisconnect} \geq  1/3$.
This concludes the proof of \cref{lem:lbstatic_disconnect}.
\end{proof}

We proceed to show that conditioned on $\attackDisconnect$, \emph{agreement} of the protocol $\pi$ is broken.
Denote by $\outmain_i$ the random variable denoting the output of $\Party_i$ according to $\mainAttack$.
Further, denote by $\JS$ the random variable corresponding to the minimal value in $\Srv$.

\begin{lemma}\label{lem:lbstatic_agreement}
$\pr{\outmain_\IS \neq \outmain_\JS \mid \attackDisconnect} \geq  1-\negl(\secParam)$.
\end{lemma}
\begin{proof}
We begin by showing that conditioned on $\attackDisconnect$, party $\Party_\Is$ outputs $0$ with overwhelming probability.
\begin{claim}\label{claim:crashA_zero}
$\pr{\outmain_\IS =0 \mid \attackDisconnect} \geq  1-\negl(\secParam)$.
\end{claim}
\begin{proof}
Consider again the adversary $\Adv_A$ that statically corrupts all parties in $\calA$ and crashes them, with the corresponding random variable $\secondAttack$.
Denote by $\crashADisconnect$ the event that $\Party_\IS$ does not send any message to parties in $\Srv$ (according to $\viewcrashA_\IS$). It holds that 
\[
\pr{\crashADisconnect} = \pr{\neg \crashAPtoS} = 1- \pr{\crashAPtoS} \geq 2/3.
\]

First, since the sender is honest and has input $0$, by \emph{validity} all honest parties in $\calB$ output $0$ in such execution, except for negligible probability. This holds even conditioned on $\crashADisconnect$ (since $\crashADisconnect$ occurs with noticeable probability). Denote by $\outcrashA_i$ the random variable denoting the output of $\Party_i$ according to $\secondAttack$. Then,
\begin{equation}\label{eq:crashA_one}
\pr{\outcrashA_\IS= 0 ~\Big|~ \crashADisconnect} \geq 1-\negl(\secParam).
\end{equation}

Second, note that conditioned on $\crashADisconnect$ (by an analogous analysis of \cref{lem:lbstatic_disconnect}, this probability is non-zero), the view of $\Party_\IS$ is identically distributed in $\secondAttack$ as its view in $\mainAttack$ conditioned on $\attackDisconnect$. Indeed, conditioned on $\attackDisconnect$, party $\Party_\IS$ receives messages only from corrupt parties in $\mainAttack$, which are consistently simulating precisely this execution where $\calA$ has crashed and the sender has input 0.  Therefore,
\begin{equation}\label{eq:crashA_two}
\pr{\outcrashA_\IS= 0 ~\Big|~ \crashADisconnect} = \pr{\outmain_\IS= 0 ~\Big|~ \attackDisconnect}.
\end{equation}

The proof follows from Equations~\ref{eq:crashA_one} and \ref{eq:crashA_two}.
This concludes the proof of \cref{claim:crashA_zero}.
\end{proof}

We proceed to show that, conditioned on $\attackDisconnect$, parties in $\Srv$ output $1$ with overwhelming probability under the attack of $\Adv_1$. Recall that $\JS$ denotes the random variable corresponding to the minimal value in $\Srv$.
\begin{claim}\label{claim:crashB_one}
$\pr{\outmain_\JS =1 \mid \attackDisconnect} \geq  1-\negl(\secParam)$.
\end{claim}
\begin{proof}
The proof follows in nearly an identical manner. Namely, consider the adversary $\Adv_B$ that statically corrupts all parties in $\calB$ and crashes them, and the random variable $\thirdAttack$.
Denote by $\crashBDisconnect$ the event that for every $J\in \Srv$, party $\Party_J$ does not send any message to $\Party_\IS$ (according to $\viewcrashB_J$).
It holds that
\[
\pr{\crashBDisconnect} = \pr{\neg \crashBStoP} = 1- \pr{\crashBStoP} \geq 2/3.
\]

Since the sender is honest and has input $1$, by \emph{validity} all honest parties in $\calA$ output $1$ except for negligible probability. This holds even conditioned on $\crashBDisconnect$ (since $\crashADisconnect$ occurs with noticeable probability). Denote by $\outcrashB_i$ the random variable denoting the output of $\Party_i$ according to $\thirdAttack$, and recall that $\JS$ corresponds to the minimal value in $\Srv$. Then,
\begin{equation}\label{eq:crashB_one}
\pr{\outcrashB_\JS= 1 ~\Big|~ \crashBDisconnect} \geq 1-\negl(\secParam).
\end{equation}

Conditioned on $\crashBDisconnect$, the view of $\Party_\JS$ is identically distributed in $\thirdAttack$ as its view in $\mainAttack$ conditioned in $\attackDisconnect$. Therefore,
\begin{equation}\label{eq:crashB_two}
\pr{\outcrashB_\JS= 1 ~\Big|~ \crashBDisconnect} = \pr{\outmain_\JS= 1 ~\Big|~ \attackDisconnect}.
\end{equation}

The proof follows from Equations~\ref{eq:crashB_one} and \ref{eq:crashB_two}.
This concludes the proof of \cref{claim:crashB_one}.
\end{proof}
Since $\Party_\Is$ and $\Party_\Js$ are honest, the proof of \cref{lem:lbstatic_agreement} follows from \cref{claim:crashA_zero} and \cref{claim:crashB_one}.
\end{proof}

Collectively, we have demonstrated an adversarial strategy $\Adv_1$ that violates the agreement property of protocol $\pi$ with  noticeable probability:
\begin{align*}
\pr{\outmain_\IS \neq \outmain_\JS}
& = \pr{\outmain_\IS \neq \outmain_\JS \mid \attackDisconnect}\cdot \pr{\attackDisconnect}\\
&  + \pr{\outmain_\IS \neq \outmain_\JS  \mid \neg \attackDisconnect}\cdot \pr{\neg \attackDisconnect}\\
& \geq \pr{\outmain_\IS \neq \outmain_\JS  \mid \attackDisconnect}\cdot \pr{\attackDisconnect}\\
& \geq (1-\negl(\secParam))\cdot \frac13.
\end{align*}
Note that the attack succeeds for any choice of distribution for setup information, and that the adversarial strategy runs in polynomial time, thus applying even in the presence of computational hardness assumptions.  This concludes the proof of \cref{thm:LB1}.
\end{proof}

%% file: lb_adaptive.tex
\section{Locality Lower Bound for Adaptive Corruptions}\label{sec:lb:adaptive}

We proceed with the proof of \cref{thm:lb:adaptive}.
Here we show how a weakly adaptive adversary that can corrupt $n/2+k$ parties can target any party of its choice and force a that party to communicate with $k$ neighbors. We refer to \cref{sec:intro:tech} for a high-level overview of the attack.

\begin{theorem}[\cref{thm:lb:adaptive}, restated]\label{thm:LB2}
\ThmLBAdaptive
\end{theorem}
\begin{proof}
Let $\pi$ be a broadcast protocol that is secure against $t=n/2+k$ adaptive corruptions.
Without loss of generality, we assume that the setup information sampled before the beginning of the protocol $(r_1,\ldots,r_n)\gets\calD_\pi$ includes the random string used by each party. That is, every party $\Party_i$ generates its messages in each round as a function of $r_i$, possibly its input (if $\Party_i$ is the sender), and its incoming messages in prior rounds.
Again, without loss of generality, let $\Party_1$ denote be the sender. Further, fix the party $\Ps$, and split the remaining parties (without $\Party_1$ and $\Ps$) to two equal-size subsets $\calS_0$ and $\calS_1$ (for simplicity, assume that $n$ is even).

Consider the following adversary $\Adv$ that proceeds as follows:
\begin{enumerate}
\item Wait for the setup phase to complete. Later on, whenever corrupting a party $\Party_i$, the adversary receive its setup information $r_i$.
\item Corrupt the sender $\Party_1$.
\item Toss a random bit $b\gets\zo$ and corrupt all parties in $\calS_{1-b}$.
\item Maintain two independent executions, denoted $\exec_0$ and $\exec_1$, as follows.
\begin{itemize}
    \item In the execution $\exec_b$, the adversary runs in its head a copy of the sender, denoted $\Party_1^b$, honestly running on input $b$ and setup $r_1$. The adversary communicates on behalf of the virtual party $\Party_1^b$, and eventually corrupted parties in $\calS_b$, with all honest parties $\calS_b\cup\sset{\Ps}$ according to this execution. The virtual parties in $\calS_{1-b}$ are emulated as crashed in this execution.

        Whenever $\Ps$ sends a message to a party $\Party_i\in\calS_b$ this party gets corrupted and ignores this message (\ie the adversary does not deliver messages from $\Ps$ to $\Party_i$).
    \item In the execution $\exec_{1-b}$, the adversary runs in its head the parties in $\calS_{1-b}$ honestly on their setup information $\sset{r_i \mid \Party_i \in\calS_{1-b}}$ and a copy of the sender, denoted $\Party_1^{1-b}$, running on input $1-b$ and setup $r_1$. The adversary communicates on behalf of the virtual parties in $(\calS_{1-b}) \cup \sset{\Party_1^{1-b}}$ with $\Ps$ according to this execution. The honest parties in $\calS_b$ are emulated as crashed in this execution; that is, the adversary drops every message sent by the virtual parties in $(\calS_{1-b}) \cup \sset{\Party_1^{1-b}}$ to $\calS_b$ and does not deliver any message from $\calS_b$ to these parties.

        Whenever $\Ps$ sends a message to a party $\Party_i\in\calS_{1-b}$ this party ignores this message (\ie the adversary does not deliver the message to $\Party_i$).
\end{itemize}
\end{enumerate}

We start by defining a few notations. Consider the following random variables
\[
\inputCoinsAdaptive=\left(R_1,\ldots,R_n,B\right),
\]
where $R_1,\ldots,R_n$ are distributed according to $\calD_\pi$, and $B$ takes a value uniformly at random in $\zo$.
During the proof, $R_i$ represents the setup information (including private randomness) of party $\Party_i$, whereas $B$ corresponds to the adversarial choice of which set to corrupt.
Unless stated otherwise, all probabilities are taken over these random variables.

Let $\mainAttack$ be the random variable defined by running the protocol $\pi$ with the adversary $\Adv$ over $\inputCoinsAdaptive$.
That is, $\mainAttack$ consists of a vector of $n+1$ views: of the parties in $\Srv_b\cup\sset{\Party_\IS}$, of the corrupted parties in $\calS_{1-b}$, where the \ith view is denoted by $\viewmain_i$, and of two copies of the sender $\Party_1^0$ and $\Party_1^1$, denoted $\viewmain_{1\mhyphen0}$ and $\viewmain_{1\mhyphen1}$, respectively. Each view consists of the setup information $R_i$, possibly the input (for the sender), and the set of received messages in each round. Specifically,
\[
\mainAttack=\left(\viewmain_{1\mhyphen0},\viewmain_{1\mhyphen1},\viewmain_2,\ldots,\viewmain_n\right).
\]

Denote by $\attackLowLocal$ the event that the output-locality of $\Party_\IS$ is at most $k$ in $\mainAttack$; that is, $\Party_\IS$ sends messages to at most $k$ parties (according to $\viewmain_\IS$).
If $\Pr[\attackLowLocal]=\negl(\secParam)$, then the proof is completed. Otherwise, it holds that $\Pr[\attackLowLocal]$ is non-negligible (in particular, $\Pr[\attackLowLocal]>0$). We will show that conditioned on $\attackLowLocal$, \emph{agreement} is broken.
Denote by $\outmain_i$ the random variable denoting the output of $\Party_i$ according to $\mainAttack$.

First, note that conditioned on $\attackLowLocal$, the view of $\Ps$ is identically distributed no matter which set $\calS_b$ is corrupted.

\begin{claim}\label{claim:adaptive_zero}
For every $\beta\in\zo$ it holds that
\[
\pr{\outmain_\is =\beta\mid \attackLowLocal \cap (B=0)} =\pr{\outmain_\is =\beta\mid \attackLowLocal \cap (B=1)}.
\]
\end{claim}
\begin{proof}
By the construction of $\Adv$, for each $\beta\in\zo$ party $\Ps$ receives from the parties in $\calS_\beta$ and from $\Party_1$ messages that correspond to an execution by honest parties on sender input $\beta$ as if the parties in $\calS_{1-\beta}$ all crashed, and where every party in $\calS_\beta$ that $\Ps$ talks to ignores its message (since $\Ps$ talks to at most $k$ parties conditioned on $\attackLowLocal$, the adversary can corrupt all of them).

Further, $\Ps$ receives from the parties in $\calS_{1-\beta}$ and from $\Party_1$ messages that correspond to a simulated execution by honest parties on sender input $1-\beta$ as if the parties in $\calS_{\beta}$ all crashed, and where every party in $\calS_{1-\beta}$ that $\Ps$ talks to ignores its message.

Clearly, the view of $\Ps$ is identically distributed in both cases; hence, its output bit is identically distributed as well.
\end{proof}

We proceed to show that conditioned on $\attackLowLocal$, party $\Ps$ outputs $0$ for $B=0$ and outputs $1$ for $B=1$.

\begin{claim}\label{claim:adaptive_one}
For every $\beta\in\zo$ it holds that
\[
\pr{\outmain_\is =\beta\mid \attackLowLocal \cap (B=\beta)} =1-\negl(\secParam).
\]
\end{claim}
\begin{proof}
Consider a different adversary for $\pi$, denoted $\Adv_\beta$, that proceeds as follows:
\begin{enumerate}
\item Wait for the setup phase to complete. 
\item Corrupt all parties in $\calS_{1-\beta}$ and crash them.
\item Whenever $\Ps$ sends a message to a party $\Party_i\in\calS_\beta$ this party gets corrupted and ignores this message (\ie the adversary does not deliver messages from $\Ps$ to $\Party_i$).
\end{enumerate}

Let $\AttackBeta$ be the random variable defined by running the protocol $\pi$ with the adversary $\Adv_\beta$ over $\inputCoins$, in which the honest sender's input is $\beta$.
That is, $\AttackBeta$ consists of a vector of $n/2$ views: of the parties in $\calS_\beta\cup\sset{\Ps}$ (both honest and corrupted), where the \ith view is denoted by $\viewcrashS_i$, and the sender $\Party_1$ denoted by $\viewcrashS_1$. Each view consists of the setup information $R_i$, the input $\beta$ for $\Party_1$, and the set of received messages in each round. Specifically,
\[
\AttackBeta=\left(\viewcrashS_i\right)_{i\in\calS_\beta\cup\sset{1,\is}}.
\]

Let us denote by $\attackLowLocalCrash$ the event that the output-locality of $\Party_\IS$ is at most $k$ in $\AttackBeta$; that is, $\Party_\IS$ sends messages to at most $k$ parties (according to $\viewcrashS_\is$).
If $\Pr[\attackLowLocalCrash]=\negl(\secParam)$, then $\Adv_\beta$ can force the locality of $\Party_\IS$ to be high in $\AttackBeta$, and the proof is completed. Otherwise, it holds that $\Pr[\attackLowLocalCrash]$ is non-negligible.

Note that since $\ssize{\calS_\beta}=(n-1)/2$ and $\locality < (n-1)/2$, then conditioned on $\attackLowLocalCrash$ there exists at least one remaining honest party in $\calS_\beta$ at the end of the execution with $\Adv_\beta$. By \emph{validity}, each such honest party must output $\beta$ with overwhelming probability. Denote by $\outcrashS_i$ the random variable denoting the output of $\Party_i$ according to $\AttackBeta$. Denote by $\JS$ the random variable corresponding to the minimal index of an honest party in $\calS_\beta$ at the end of the execution with $\Adv_\beta$. Then
\begin{equation}\label{eq:adaptive:one}
\pr{\outcrashS_\JS =\beta\mid \attackLowLocalCrash} =1-\negl(\secParam).
\end{equation}
Further, note that the set of all honest parties in $\calS_\beta$ and their joint view in an execution with $\Adv_\beta$ conditioned on $\attackLowLocalCrash$ is identically distributed as in an execution with $\Adv$ conditioned on $\attackLowLocal\cap (B=\beta)$. Therefore,
\begin{equation}\label{eq:adaptive:two}
\pr{\outcrashS_\JS =\beta\mid \attackLowLocalCrash} =\pr{\outmain_\JS =\beta\mid \attackLowLocal\cap (B=\beta)}.
\end{equation}
Finally, by \emph{agreement}, since both $\Party_\JS$ and $\Ps$ are honest at the end of the execution with $\Adv$, conditioned on $\attackLowLocal\cap (B=\beta)$, it holds that
\begin{equation}\label{eq:adaptive:three}
\pr{\outmain_\is =\beta\mid \attackLowLocal\cap (B=\beta)} =\pr{\outmain_\JS =\beta\mid \attackLowLocal\cap (B=\beta)}-\negl(\secParam).
\end{equation}
The claim follows from Equations \ref{eq:adaptive:one}, \ref{eq:adaptive:two}, and \ref{eq:adaptive:three}.
\end{proof}

By \cref{claim:adaptive_zero} and \cref{claim:adaptive_one} it follows that $\Pr[\attackLowLocal]=\negl(\secParam)$.
This concludes the proof of \cref{thm:LB2}.
\end{proof}

%% file: flood.tex
\newcommand{\type}{\mathsf{T}}
\newcommand{\neighborhood}{\mathcal{N}}
\newcommand{\committee}{\mathcal{C}}
\newcommand{\relayed}{\mathsf{Relayed}}

\section{Statically Secure Sub-Quadratic Broadcast}\label{sec:ub}

In this section we prove \cref{thm:ub:static} by presenting a broadcast protocol secure against a constant fraction of static corruptions that requires $\tilde{O}(n)$ bits of total communication, given a trusted-PKI setup for VRFs. The protocol is balanced, and each party communicates $\polylog(n)\cdot\poly(\secParam)$ bits.

\begin{proposition}[\cref{thm:ub:static}, restated]\label{thm:UB}
\ThmUBStatic
\end{proposition}

Our protocol is a simple variant of Chan, Pass, and Shi \cite{CPS20}, where every step that requires all-to-all communication is substituted by a more communication-efficient message-propagation mechanism. We first describe the message-propagation mechanism, which follows the spirit of Tsimos, Loss and Papamanthou \cite{TLP22}, and afterwards the modified broadcast protocol of \cite{CPS20}.

\paragraph{A message-propagation mechanism.}\label{sec:flood}

Consider the problem where each party has (possibly) an input message it would like to disseminate to all the other parties. This can trivially be solved by letting each party send its input message to all other parties, which would incur a communication complexity of $n^2 \cdot \ell$, where $\ell$ is the length of the input message.

In some cases, however, we are interested in disseminating only messages of a certain type. More precisely, we want that if any honest party has a message of type $\type$, then all honest parties obtain at least one message of type $\type$, but we do not need that all parties obtain all input messages of type $\type$.

\begin{definition}[Message-propagation protocol]\label{def:flooding}
Let $\type$ be a predicate. An $n$-party protocol $\pi$, where each party $\Party_i$ has an initial input $x_i$ (or no input), is a \emph{message-propagation} protocol for type-$\type$ messages secure against $t$ static corruptions, if for any PPT adversary that statically corrupts up to $t$ parties, the following holds except for negligible probability in $\secParam$: If any honest party holds an input of type-$\type$, then all honest parties output a value of type-$\type$.
\end{definition}

The trivial approach described above (where every party sends its input message to every other party) still requires quadratic communication, since a linear number of parties may distribute a type-$\type$ message towards all parties. A more efficient solution employs instead a flooding mechanism over a communication graph that forms an expander (i.e., a sparse graph with strong connectivity properties). More concretely, it is possible to form a communication graph where each party is connected only to $O(\log(n) + \secParam)$ other parties, and the honest parties form a connected component except with negligible probability in $\secParam$.

Each party $\Party_i$ can then send its input to its neighbors (if the input is of type $\type$), and if $\Party_i$ has no input or an input that is not of type $\type$, party $\Party_i$ can simply forward to its neighbors the first type-$\type$ message that it received. It is easy to see that, since the honest parties form a connected component, if any honest party has a message of type $\type$, then all honest parties obtain at least one message of type $\type$. Moreover, the communication complexity is $O(n\cdot (\log(n) + \secParam) \cdot \ell)$, where $\ell$ is the length of a type-$\type$ message. We describe the protocol in \cref{prot:flood1} and obtain the following lemma.

\pprotocol{$\flood(\type, n, \epsilon, \secParam)$}{Message-propagation mechanism for $n$ parties for messages of type $\type$.}{prot:flood1}{thb}{

\textbf{Parameters:} 
$\type$ is the type of messages to propagate, $n$ is the number of parties, $\epsilon$ is the fraction of honest parties, and $\secParam$ is the security parameter.
\smallskip

\textbf{Variables:}  
$\Party_i$ sets local variables $\neighborhood_i = \emptyset$ and $\relayed_i = 0$.
\smallskip

\textbf{Message propagation:}

Each party $\Party_i$ has an input $x_i$ (no input is interpreted as $\bot$).

Each party $\Party_i$ locally adds $\Party_j$ to the set $\neighborhood_i$ with uniform probability $p_\mathsf{flood}\assign\frac{\log(n)+\secParam}{\epsilon n}$.\footnote{We note that with static corruptions, each party $\Party_i$ can maintain the same neighborhood set $\neighborhood_i$ across several instances of $\Flood$.}

Let $\rho=\floodrounds$. Each party $\Party_i$ does the following.

Round $1$: If the input message $x_i$ is of type $\type$, send $x_i$ to each $\Party_j\in \neighborhood_i$, set $\relayed_i = 1$, and set $y_i = x_i$.

\For{each round $r\in\sset{2,\ldots,\rho}$}{
    Let $\calS$ be the messages received in round $r-1$. If there is a message $m \in \calS$ of type-$\type$ and $\relayed_i = 0$, then send $m$ to each $\Party_j\in \neighborhood_i$, set $\relayed_i = 1$, and set $y_i = m$.
}

Output $y_i$.
}

\begin{lemma}\label{lem:flood}
Let $\secParam$ be a security parameter, let $n$ be the number of parties, and let $0<\epsilon<1$ be a constant. Protocol $\flood(\type,n,\epsilon,\secParam)$ is a message-propagation protocol for type-$\type$ messages, secure against static $(1-\epsilon)\cdot n$ corruptions. The communication complexity is $O(n\cdot (\log(n) + \secParam) \cdot \ell)$ bits, where $\ell$ is the size of a type-$\type$ message.
\end{lemma}

\begin{proof}
From \cite[Lem.\ 15]{AC:LMMRT22}, we know that the communication graph induced by the honest parties during an execution of $\flood$ is connected and has diameter at most $\rho=\floodrounds$ with overwhelming probability in $\secParam$. Therefore, if any honest party has a type-$\type$ input message, all honest parties receive a type-$\type$ message within $\rho$ rounds.

Since each honest party only sends a type-$\type$ message at most once, and the neighborhood of a party is of size $O(\log(n) + \secParam)$, the claimed communication complexity follows.
\end{proof}

\paragraph{Chan et al.'s modified protocol.}\label{sec:floodbroadcast}

We describe a modified version of Chan et al.'s broadcast protocol, where every all-to-all communication step is simply substituted by the message-propagation mechanism described above.

Following \cite{ACDNPRS19,CPS20}, we describe the protocol in a hybrid world assuming an ideal functionality $\Fmine$ parameterized by  probability $p \assign \min\{1,\frac{\kappa+1}{\epsilon n}\}$, which can be realized assuming a trusted-PKI for VRFs as setup. $\Fmine$ serves as a committee-election oracle, and has the following interface:
\begin{itemize}
    \item Mining: when a party $\Party_i$ calls $\Fmine.\mine(b)$ on a bit $b\in\{0,1\}$ for the first time, $\Fmine$
    flips a $p$-weighted coin and returns the result $b'\in\{0,1\}$. (1 indicates success, 0 indicates failure.)
    Calling $\Fmine.\mine(b)$ again in the future returns the same result $b'$.
    \item Verifying: any node can call $\Fmine.\verify(b,i)$ on a bit $b$ and index $i$. If $\Party_i$ has already called $\Fmine.\mine(b)$ and received result $b'=1$, then $\Fmine.\verify(b,i)$ returns 1; otherwise, $\Fmine.\verify(b,i)$ returns 0.
\end{itemize}

At the start of the protocol, parties invoke the ideal functionality $\Fmine$ to choose a committee $\committee$ of size $\ssize{\committee}=R = O(\secParam)$ parties, which will contain at least one honest party (with high probability). The protocol is based on the Dolev-Strong protocol \cite{DS83}, where only the committee members $\committee$ (and the sender) contribute signatures on the input message and every protocol stage $r$ is split into two mini-stages. We define a message of type-$\type_{b,r}$ as a message that contains the bit $b$ along with at least $r$ correct signatures, including the signatures from the sender and at least $r-1$ distinct committee-parties.

\begin{itemize}
    \item Each party $\Party_i$ keeps a set $\ext_i$ (initially empty). In stage $0$, the sender signs its input bit and sends $b$ and the signature to all parties (this is a $\type_{b,1}$ message).
    \item For each stage $r = 1$ to $R+1$, each party $\Party_i$ does the following.
    \begin{enumerate}
        \item First mini-stage: for every $b \notin \ext_i$ such that $\Party_i$ has received a type-$\type_{b,r}$ message, add~$b$ to $\ext_i$ and run the message-propagation protocol $\flood(\type_{b,r},n,\epsilon, \secParam)$ to distribute a type-$\type_{b,r}$ to all parties (and wait until the protocol ends).
        \item Second mini-stage: for every $b \notin \ext_i$ such that a committee member $\Party_i \in \committee$ has received a type-$\type_{b,r}$ message, pick any such message, add $b$ to $\ext_i$, create a type-$\type_{b,r+1}$ message by adding its own signature, and run the message-propagation protocol $\flood(\type_{b,r+1},n,\epsilon, \secParam)$ to distribute a type-$\type_{b,r+1}$ message to all parties (and wait until the protocol ends).
    \end{enumerate}
    \item Stage $R+2$: Each party $\Party_i$ outputs the bit contained in $\ext_i$ if $|\ext_i| = 1$; otherwise, output~$0$.
\end{itemize}

See \cref{prot:cps20+gossip} for the formal description of the protocol.
\input{prot_CPS20}
We proceed to prove \cref{thm:UB}.

\begin{proof}[Proof of \cref{thm:UB}]
Termination is trivial, since all parties output a value at stage $R+2$.

To prove validity, let the sender be an honest party with input $b$. At the end of stage~$0$, all honest parties receive a type-$\type_{b,1}$ message. Therefore, each party $\Party_i$ adds $b$ to the set $\ext_i$ in the first mini-stage of stage $1$. Moreover, no other value is added to $\ext_i$ (except for negligible probability), since the sender does not sign any other value. At stage $R+2$, all honest parties output $b$ except for negligible probability.

To prove agreement, we show that the sets $\ext_i$ of each honest party $\Party_i$ contain the same set of values the end of the protocol at stage $R+2$.
\begin{itemize}
\item
First, consider the case where the first honest party $\Party_i$ that adds a bit $b$ to its set $\ext_i$, does so in the first mini-stage of stage $r$. This means that $\Party_i$ received a type-$\type_{b,r}$ message and a type-$\type_{b,r}$ message was subsequently propagated through $\Flood$. Note that since a type-$\type_{b,r}$ message contains $r$ signatures, and $\committee$ contains at least one honest party, then it holds that $r < R+1$ (otherwise an honest committee-member added $b$ to its extracted set previously, contradicting the assumption that $\Party_i$ is the first honest party that adds $b$ to its set).

The message-propagation mechanism then ensures that all honest parties receive a type-$\type_{b,r}$ message at the end of this first mini-stage. Therefore, in the second mini-stage of stage~$r$, there is an honest committee-member that can form a type-$\type_{b,r+1}$ message, which is distributed through $\Flood$. Therefore, in the next stage $r+1$ all parties add $b$ to their extracted set.
\item
Second, consider the case where the first honest party $\Party_i$ that adds a bit $b$ to its set $\ext_i$, does so in the second mini-stage of stage $r$. Since the second mini-stage is only executed by committee members, $\Party_i \in \committee$. Moreover, $\Party_i$ received a type-$\type_{b,r}$ message for the first time and can form a type-$\type_{b,r+1}$ message, which is distributed through $\Flood$. Note that $r < R+1$ since the type-$\type_{b,r}$ message does not contain $\Party_i$'s signature. Therefore, in the next stage $r+1$ (first mini-stage), all parties add $b$ to their extracted set.
\end{itemize}

In each stage $r$, there are at most $4$ invocations to $\Flood$ (one per bit, per mini-stage), and the message contains up to $R+1$ signatures and a bit value. Assuming each signature is of size $O(\secParam)$, the size of the message is bounded by $\ell = O((R+1)\secParam + 1) = O(\secParam^2)$. Since the number of stages invoking $\Flood$ is $R+1 = O(\secParam)$ and the cost of each instance of $\Flood$ is $O(n\cdot(\log(n) + \secParam) \cdot \ell)$, the total incurred communication complexity is $O(\secParam \cdot n\cdot(\log(n) + \secParam) \cdot \secParam^2) = O(\secParam^3 n \log(n) + \secParam^4  n))= \tilde{O}(n)$ and the per-party communication is $\tilde{O}(1)$.
\end{proof}

%% file: prot_CPS20.tex
\pprotocol{$\fbc$}{Chan et al.'s modified Broadcast Protocol in the $\Fmine$-hybrid world}{prot:cps20+gossip}{!htb}{
\begin{description}
\item [Parameters:]
Let $n$ be the number of parties, let $\epsilon$ be the fraction of honest nodes, and let $\secParam$ the security parameter. Let $(\KeyGen, \Sign, \Ver)$ be a digital signature scheme.

\item [Setup:]
Let $\Fmine$ be the committee-election oracle. The parties have access to a PKI setup where each party $\Party_i$ samples $(\pk_i,\sk_i)\gets\KeyGen(1^\secParam)$ and all parties obtain $(\pk_1,\ldots,\pk_n)$.

\item [Input:]
Let $b_s$ be the input of the sender $\Party_s$.

\item [Stage $0$ (Initialization):]~
\begin{itemize}[leftmargin=*]
    \item The sender $\Party_s$ computes $\sigma=\Sign_{\sk_s}(b_s)$ and sends the message $(b_s,\sigma)$ to all parties.
    \item Each party $\Party_i$ sets its extracted set $\ext_i\gets\emptyset$ and also queries $\Fmine$, parameterized with $p\assign\min\{1,\frac{\secParam+1}{\epsilon n}\}$. Let $\committee$ denote the elected committee.
\end{itemize}

Let the number of stages be $R\assign\frac{3}{\epsilon}(\secParam+1)$. In the following, we denote by type-$\type_{b,r}$ as a message that contains the bit $b$ along with at least $r$ correct signatures (\ie such that $\Ver$ outputs $1$) from public keys of distinct committee-parties and also including the sender.

\item [Stage $r=1,\ldots, R+1$:]~
\begin{enumerate}[leftmargin=*]
\item (Stage $r.1$) Each party $\Party_i$ does the following:
\begin{itemize}[leftmargin=*]
    \item For each bit $b$ such that $\Party_i$ has received a type-$\type_{b,r}$ message $m$ and $b\notin\ext_i$: add $b$ to $\ext_i$ and invoke protocol $\flood(\type_{b,r},n, \epsilon, \secParam)$, to propagate one type-$\type_{b,r}$ message $m$.
\end{itemize}

\item (Stage $r.2$) Each party $\Party_i\in \committee$ does the following:
\begin{itemize}[leftmargin=*]
\item For each bit $b$ such that $\Party_i$ has received a type-$\type_{b,r}$ message $m$ and $b\notin\ext_i$: add $b$ to~$\ext_i$ and if $\Party_i$ can create a type $\type_{b,r+1}$ message $m$ by adding a signature $\Sign_{\sk_i}(b)$ to the received $\type_{b,r}$ message, execute protocol $\flood(\type_{b,r+1},n, \epsilon, \secParam)$, to propagate one type-$\type_{b,r+1}$ message $m$.
\end{itemize}
\end{enumerate}

\item [Stage $R+2$ (Termination):]
Each party $\Party_i$ does the following:
\begin{itemize}[leftmargin=*]
    \item If $|\ext_i|=1$, then $\Party_i$ outputs the unique bit $b_i\in\ext_i$; else, $\Party_i$ outputs the default bit~$0$.
\end{itemize}
\end{description}
}